\documentclass[letter,UKenglish,cleveref, autoref]{oasics-v2019}
\usepackage{algorithm2e}
\usepackage{algorithmic}
\usepackage{amsmath}
\DeclareMathOperator*{\argmax}{arg\,max}
\DeclareMathOperator*{\argmin}{arg\,min}
\title{The SBP Algorithm for Maximizing Revenue in Online Dial-a-Ride} %TODO

\titlerunning{}%optional, please use if title is longer than one line

\author{Ananya Christman, Nicholas Jaczko, Tianzhi Li, Scott Westvold, Xinyue Xu }{Middlebury VT 05753, USA }{achristman@middlebury.edu}{https://orcid.org/0000-0001-9445-1475}{}%%TODO mandatory, please use full name; only 1 author per \author macro; first two parameters are mandatory, other parameters can be empty. Please provide at least the name of the affiliation and the country. The full address is optional

\author{Christine Chung}{Connecticut College, New London CT 06320, USA }{cchung@conncoll.edu}{https://orcid.org/0000-0003-3580-9275}{}%

\authorrunning{A D Christman et. al}%TODO mandatory. First: Use abbreviated first/middle names. Second (only in severe cases): Use first author plus 'et al.'
\Copyright{Ananya Christman and Christine Chung}%TODO mandatory, please use full first names. LIPIcs license is "CC-BY";  http://creativecommons.org/licenses/by/3.0/

%TODO mandatory; please add comma-separated list of keywords

\category{}%optional, e.g. invited paper

\relatedversion{}%optional, e.g. full version hosted on arXiv, HAL, or other respository/website
\supplement{}%optional, e.g. related research data, source code, ... hosted on a repository like zenodo, figshare, GitHub, ...

\nolinenumbers %uncomment to disable line numbering

\EventEditors{}
\EventNoEds{2}
\EventLongTitle{}
\EventShortTitle{}
\EventAcronym{}
\EventYear{}
\EventDate{}
\EventLocation{}
\EventLogo{}
\SeriesVolume{}
\ArticleNo{}
\newcommand{\on}{\textsc{on}}
\newcommand{\opt}{\textsc{opt}}
\newcommand{\sbp}{\textsc{sbp}}
\newcommand{\sbpprime}{\textsc{sbp}'}
\newcommand{\sbpdb}{\textsc{sbp}''}
\newcommand{\sbpbar}{\overline{\textsc{sbp}}}
\newcommand{\optbarprime}{\overline{\textsc{opt}}}

\newcounter{comment}[section]

\begin{document}

\maketitle

\begin{abstract}
\emph{
In the Online-Dial-a-Ride Problem (OLDARP) a server travels through a metric space to serve requests for rides. We consider a variant where each request specifies a source, destination, release time, and revenue that is earned for serving the request. The goal is to maximize the total revenue earned within a given time limit. We  prove that no non-preemptive deterministic online algorithm for OLDARP can be guaranteed to earn more than twice the revenue earned by an optimal offline solution. We then investigate the  \textsc{segmented best path} ($\sbp$) algorithm of~\cite{atmos17} for the general case of weighted graphs. The previously-established lower and upper bounds for the competitive ratio of $\sbp$ are 4 and 6, respectively, under reasonable assumptions about the input instance. We eliminate the gap by proving that the competitive ratio is 5 (under the same reasonable assumptions). %We also prove that no non-preemptive deterministic online algorithm for OLDARP can be guaranteed to earn any fraction of the revenue earned by an optimal offline solution in the last $T-2T/f$ time units.  
 %in the first $T-2T/f$ time units.  \adcnote{We may need to rephrase the abstract since we are now mentioning the general lower bound so we don't have to say "reasonable assumptions about the input" - we could just say what the assumptions are. Or we could just remove this text about the general lower bounds}. 
We also prove that when revenues are uniform, $\sbp$ has competitive ratio 4. %(again, under reasonable assumptions about the input instance). 
 Finally, we provide a competitive analysis of $\sbp$ on complete bipartite graphs.}
\end{abstract}

\newpage

\section{ Introduction}
In the On-Line Dial-a-Ride Problem (OLDARP), a server travels in some metric space to serve requests for rides. Each request specifies a \textit{source}, which is the pick-up (or start) location of the ride, a \textit{destination}, which is the delivery (or end) location, and the \textit{release time} of the request, which is the earliest time the request may be served. Requests arrive over time; specifically, each arrives at its release time and the server must decide whether to serve the request and at what time, with the goal of meeting some optimality criterion. The server has a \textit{capacity} that specifies the maximum number of requests it can serve at any time.  Common optimality criteria include minimizing the total travel time (i.e. makespan) to satisfy all requests, minimizing the average completion time (i.e. latency),  or maximizing the number of served requests within a specified time limit.  In many variants \textit{preemption} is not allowed, so if the server begins to serve a request, it must do so until completion. On-Line Dial-a-Ride Problems have many practical applications in settings where a vehicle is dispatched to satisfy requests involving pick-up and delivery of people or goods. Important examples include ambulance routing, transportation for the elderly and disabled, taxi services including Ride-for-Hire systems (such as Uber and Lyft), and courier services.

We study a variation of OLDARP where in addition to the source, destination and release time, each request also has a priority and there is a time limit within which requests must be served. The server has unit capacity and the goal for the server is to serve requests within the time limit so as to maximize the total priority. A request's priority may simply represent the importance of serving the request in settings such as courier services. In more time-sensitive settings such as ambulance routing, the priority may represent the urgency of a request. In profit-based settings, such as taxi and ride-sharing services, a request's priority may represent the revenue earned from serving the request.  For the remainder of this paper, we will refer to the priority as ``revenue,'' and to this variant of the problem as ROLDARP. Note that if revenues are uniform the problem is equivalent to maximizing the number of served requests. %This variant is useful for delivery services where all requests have equal priorities, e.g for not-for-profit services that provide transportation to elderly and disabled passengers and courier services where deliveries are not prioritized.

\subsection{ Related work}

\label{related}%\ccnote{I have cut from this section to save some space... we could cut it even further...}
The Online Dial-a-Ride problem was introduced by Feuerstein and Stougie~\cite{stougie} and several variations of the problem have been studied since.  For a comprehensive survey on these and many other problems in the general area of \textit{vehicle routing} see~\cite{wagnersurvey} and~\cite{typologysurvey}. Feuerstein and Stougie studied the problem for two different objectives:  minimizing completion time and minimizing latency.  For minimizing completion time, they showed that any deterministic algorithm must have competitive ratio of at least 2 regardless of the server capacity. They presented algorithms for the cases of finite and infinite capacity with competitive ratios of 2.5 and 2, respectively. For minimizing latency, they proved that any algorithm must have a competitive ratio of at least 3 and presented a 15-competitive algorithm on the real line when the server has infinite capacity. Ascheuer et al.~\cite{ascheuer} studied OLDARP with multiple servers with the goal of minimizing completion time and presented a 2-competitive algorithm. %Jaillet and Wagner~\cite{offline2} considered a version of OLDARP where each request consists of one or more locations and precedence and capacity requirements for the locations and presented a nonpolynomial-time 2-competitive algorithm for minimizing completion time.  
More recently,  Birx and Disser~\cite{birx} studied OLDARP on the real line and presented a new upper bound of 2.94 for the \textsc{smartstart} algorithm~\cite{ascheuer}, which improves the previous bound of 3.41\cite{krumke00}. For OLDARP on the real line, Bjelde et al.~\cite{bjelde} present a preemptive algorithm with competitive ratio 2.41.

The Online Traveling Salesperson Problem (OLTSP), introduced by Ausiello et al.~\cite{ausiello} and also studied by Krumke~\cite{krumke}, is a special case of OLDARP where for each request the source and destination are the same location.  %Krumke~\cite{krumke} studied both OLDARP and OLTSP for the uniform metric space with the objective of minimizing the maximum \textit{flow time}, that is the difference between a request's release and service times. They proved that no competitive algorithm exists for OLDARP and gave a 2-competitive algorithm to solve OLTSP.  
There are many studies of variants of OLDARP and OLTSP~\cite{ausiello,jaillet-wagner,jawgal,krumke} that differ from the variant that we study which we omit here due to space limitations.
%More recently, Jawgal et al.~\cite{jawgal} studied a special case of OLTSP in which requests are located only on the circumference of a circle and the server moves only along the circumference of that circle. %For the \textit{homing} variant of the problem, the objective is to minimize the time to return to the origin after serving
%all the requests. %For the \textit{nomadic} variant, after serving all the requests, the server is not required to return to the origin and the objective is to minimize the time to serve all requests.  %For both variants, the authors considered \textit{zealous} algorithms, i.e. for which the server does not wait when there are unserved requests. For the nomadic variant, they we proved a lower bound of \frac{28}{13} on the competitive ratio of any zealous online algorithm
 %They present a 2.5-competitive online algorithm for  the objective of minimizing the total time to serve all requests.
  %Ausiello et al.~\cite{ausiello} studied the related Online Prize Collecting Traveling Salesman Problem (PCTSP), where the server visits a set of cities where cities arrive over time and each city has a given prize and penalty.  The authors presented a 7/3-competitive algorithm for this problem when the goal is to collect a given quota of prizes of cities while minimizing the length of the tour plus the penalties of the cities not in the tour.

In this paper, we study OLDARP where each request has a revenue that is earned if the request is served and the goal is to maximize the total revenue earned within a specified time limit; the offline version of the problem was shown to be NP-hard in~\cite{atmos17}.  More recently, it was shown that even the special case of the offline version with uniform revenues and uniform weights is NP-hard~\cite{twochain}.    
Christman and Forcier~\cite{christman1} 
%showed that no deterministic algorithm can be competitive for this problem when edge weights are nonuniform and when revenues can take on any arbitrarily large value. They therefore focused on graphs with uniform edge weights and 
presented a 2-competitive algorithm for OLDARP on graphs with uniform edge weights. Christman et al.~\cite{atmos17} showed that the
%non-competitiveness
lack of a competitive algorithm for OLDARP with nonuniform edge weights is due to arbitrarily large edge weights alone, i.e. if edge weights may be arbitrarily large, then regardless of revenue values, no deterministic algorithm can be competitive.  They therefore considered graphs where edge weights are bounded by $T/f$, where $T$ is the time limit, for some $1 < f <T$, and gave a 6-competitive algorithm for this problem. Note that this is a natural subclass of inputs since in real-world dial-a-ride systems, drivers would be unlikely to spend a large fraction of their day moving to or serving a single request.

%here 6/18

\subsection{ Our results}

In this work we begin with improved lower and upper bounds for the competitive ratio of the \\ $\textsc{segmented best path}$ ($\sbp$) algorithm that was presented in~\cite{atmos17}. %\adcnote{modified} We study $\sbp$ because it has the best known competitive ratio for ROLDARP and is a relatively straightforward algorithm.\ccnote{sounds good to me} %and a natural extension of the \textsc{greatest revenue first} (\textsc{grf}) algorithm of~\cite{christman1}, which yields the best known competitive ratio for ROLDARP on the uniform metric space. 
In~\cite{atmos17}, it was shown that $\sbp$'s competitive ratio  has lower bound 4 and upper bound 6, provided that the edge weights are bounded by $T/f$ where $T$ is the time limit and $1 < f < T$, and that the revenue earned by the optimal offline solution in the last $2T/f$ time units is bounded by a constant. This assumption is imposed because, 
%$\sbp$ is unable to earn this revenue.
%for reason similar to why no algorithm can be guaranteed to earn the revenue of the final requests served by $\opt$. 
as we show in Lemma~\ref{lem:last2},
%we substantiate $\sbp$'s inability to earn this revenue by showing that
\textit{no} non-preememptive deterministic online algorithm can be guaranteed to earn this revenue.  We also show that \textit{no} non-preemptive deterministic online algorithm for OLDARP can be guaranteed to earn more than twice the revenue earned by an optimal offline solution in the first $T-2T/f$ time units. 
 We then close the gap between the upper and lower bounds of $\sbp$ by providing an instance where the lower bound is 5 (Section~\ref{sec:lb}) and a proof for an upper bound of 5 (Section~\ref{sec:ub}). We note that another interpretation of our result is that under a weakened-adversary model where $\opt$ has two fewer time segments available, while $\sbp$ has the full time limit $T$, $\sbp$ is 5-competitive. We then investigate the problem for uniform revenues (so the objective is to maximize the total number of requests served) and prove that $\sbp$ earns at least 1/4 the revenue of the optimal solution, minus an additive term linear in $f$, the number of time segments (Section~\ref{sec:uniformrev}). % competitive ratio 4. 
This variant is useful for settings where all requests have equal priorities such as not-for-profit services that provide transportation to elderly and disabled passengers and courier services where deliveries are not prioritized. 

We then consider the problem for complete bipartite graphs; for these graphs every source is from the left-hand side and every destination is from the right-hand side (Section \ref{sec:bipartite}). These graphs model the scenario where only a subset of locations may be source nodes and a disjoint subset may be destinations, e.g. in the delivery of goods from commercial warehouses only the warehouses may be sources and only customer locations may be destinations.  
We refer to this problem as ROLDARP-B. We first show that if edge weights are not bounded by a minimum value, then ROLDARP on general graphs reduces to ROLDARP-B. We therefore impose a minimum edge weight of $kT/f$ for some constant $k$ such that $0<k\leqslant 1$. We show that if revenues are uniform, $\sbp$ has competitive ratio $\lceil 1/k \rceil$. Finally, we show that if revenues are nonuniform $\sbp$ has competitive ratio $\lceil 1/k \rceil$, provided that the revenue earned by the optimal offline solution in the last $2T/f$ time units is bounded by a constant. %\adcnote{should we mention Lemma~\ref{no2tf} here?}\ccnoteP{like this?} 
(This assumption is justified by Lemma \ref{lem:last2} which says no deterministic algorithm can be guaranteed to earn any fraction of what is earned by the optimal solution in the last $2T/f$ time units.)

\section{ Preliminaries}
%read

The Revenue-Online-Dial-a-Ride Problem (ROLDARP) is formally defined as follows. The input is an undirected complete graph $G = (V, E)$ where $V$ is the set of vertices (or nodes) and $E = \{(u, v) : u, v \in V, u \neq v \}$ is the set of edges.  For every edge $(u, v) \in E$, there is a weight $w_{u, v} > 0$, which represents the amount of time it takes to traverse $(u, v)$. We note that any simple, undirected, connected, weighted graph is allowed as input, with the simple pre-processing step of adding an edge wherever one is not present whose weight is the length of the shortest path between its two endpoints.  We further note that the input can be regarded as a metric space if the weights on the edges are expected to satisfy the triangle-inequality.  One node in the graph, $o$, is designated as the origin and is where the server is initially located (i.e. at time 0). The input also includes a time limit $T$ and a sequence of requests, $\sigma$, that are dynamically issued to the server. 

Each request is of the form $(s, d, t, p)$ where $s$ is the source node, $d$ is the destination, $t$ is the time the request is released, and $p$ is the revenue (or priority) earned by the server for serving the request. The server does not know about a request until its release time $t$.  To serve a request, the server must move from its current location $x$ to $s$, then from $s$ to $d$. The total time for serving the request is equal to the length of the path from $x$ to $d$ and the earliest time a request may be released is at $t=0$. For each request, the server must decide whether to serve the request and if so, at what time. A request may not be served earlier than its release time and at most one request may be served at any given time. Once the server decides to serve a request, it must do so until completion.  The goal for the server is to serve requests within the time limit so as to maximize the total earned revenue. 

The authors of \cite{atmos17} showed that if edge weights may be arbitrarily large then no deterministic algorithm can be competitive. They therefore considered graphs where edge weights are bounded by $T/f$ where $T$ is the time limit, for some $1 < f <T$, and presented the \textsc{segmented best path} ($\sbp$) algorithm for this problem (please refer to Algorithm \ref{sbpalg}).  %However, the authors showed that even in this setting, no deterministic online algorithm can serve the requests served by \textsc{opt} during the last $T/f$ time units.  They presented a 6-competitive algorithm, \textsc{segmented best path} ($\sbp$) that has an additive factor equal to the revenue earned by $\opt$ during the last $2T/f$ time units. 

\RestyleAlgo{boxruled}
\begin{algorithm}\caption{Algorithm \textsc{Segmented Best Path (\sbp)}. Input is complete graph $G$ with time limit $T$ and maximum edge weight $T/f$.}
\label{sbpalg}
\begin{algorithmic} [1]
\STATE Let $t_1, t_2, \ldots t_{f}$ denote the time segments ending at times $T/f, 2T/f, \ldots, T$, resp.%ectively.

% \IF {$f$ is odd}
% \STATE At $t_1$, do nothing.
%  \STATE At the start of every $t_i$ for even $i \ge 2$ find the \textit{max-revenue-request-set} and move \\to the source location of the first request in this set. \\Denote this \textit{request-set} as $R$. \\If no unserved  request-sets exist, do nothing.

%  \STATE At the start of every $t_i$ for odd $i \ge 3$, serve request-set $R$ (if it exists) from the previous step.
%  \ENDIF

% \IF {$f$ is even}
% \STATE At the start of every $t_i$ for odd $i \ge 1$ find the \textit{max-revenue-request-set} and move \\ to the source location of the first request in this set. \\Denote this request-set as  $R$. \\If no unserved request sets exist, do nothing.

% \STATE At the start of every $t_i$ for even $i \ge 2$, serve request-set $R$ (if it exists) from the previous step.

% \ENDIF
\STATE Let $i = 1$.
\IF {$f$ is odd}
\STATE At $t_1$, do nothing. Increment $i=2$.
%\ELSE
%\STATE Set $i=1$.
\ENDIF
\WHILE {$i < f$}
\STATE At the start of $t_i$, find the \textit{max-revenue-request-set}, $R$. 
\IF{$R$ is non-empty}
\STATE Move to the source location of the first request in $R$.
%\STATE Let $i=i+1$.
\STATE At the start of $t_{i+1}$, serve request-set $R$.% from the previous step. 
\ELSE 
\STATE Remain idle for $t_i$ and $t_{i+1}$
\ENDIF
\STATE Let $i=i+2$.
\ENDWHILE
\end{algorithmic}
\end{algorithm}

The algorithm $\sbp$ starts by splitting the total time $T$ into $f$ segments each of length $T/f$.
At the start of a time segment, the server determines the \textit{max-revenue-request-set}, i.e. the  maximum revenue set of requests that can be served within one time segment, and moves to the source of the first request in this set. During the next time segment, it serves the requests in this set.  It continues this way, alternating between moving to the source of first request in the max-revenue-request-set during one time segment, and serving this request-set in the next time segment. To find the max-revenue-request-set, the algorithm maintains a directed auxiliary graph, $G'$ to keep track of unserved requests (an edge between two vertices $u$,$v$ represents a request with source $u$ and destination $v$). It finds all paths of length at most $T/f$ between every pair of nodes in $G'$ and returns the path that yields the maximum total revenue (please refer to~\cite{atmos17} for full details). Finding all paths of length at most $T/f$ in $G'$ requires enumeration of all paths in $G'$ and the number of possible paths is exponential in the size of $G'$, which is determined directly by the number of outstanding requests in the current time segment.  However, in many real world settings, the size of $G'$ will be small relative to the size of $G$ and in settings where $T/f$ is small, the run time is further minimized. Therefore it should be feasible to execute the algorithm efficiently in many realistic settings. 

	It was observed in \cite{atmos17} that no deterministic online algorithm can be guaranteed to serve the requests served by \textsc{opt} during the last time segment and the authors proved that $\sbp$ is 6-competitive barring an additive factor equal to the revenue earned by $\opt$ during the last two time segments.  More formally, let $rev(\sbp (t_j))$ and 
	$rev(\opt(t_j))$ denote the revenue earned by $\sbp$ and $\opt$ respectively during the $j$-th time segment.  Then if $rev(\opt({t_f})) + rev(\opt({t_{f-1}})) \leq c$ for some constant $c$, then
$\sum_{j=1}^f rev(\opt({t_j})) \le 6 \sum_{j=1}^f rev(\sbp({t_j}))  + c$. 
It was also shown in \cite{atmos17} that as $T$ grows, the competitive ratio of $\sbp$ is at best
 4 (again with the additive term equal to $rev(\opt({t_f})) + rev(\opt({t_{f-1}})))$, resulting in a gap between the upper and lower bounds. 

%\adcnote{added section below}
\subsection{ General lower bound}
\label{sec:generalLB}
%\adcnote{Also, seems like this blurb before this lemma statement and the blurb before the next one are just repeating the lemma statements?} \ccnote{okay, rewording}
%\adcnote{Again, I think this instance also works for preemptive algorithms...?}
We first show that \textit{no} non-preemptive deterministic online algorithm (e.g. $\sbp$) can be competitive with the revenue earned by an optimal offline solution in the last two segments of time. We note that this claim applies to a stronger notion of non-preemption where, as in real-world systems like Uber/Lyft, once the server decides to serve a request, it must move there and serve it to completion. %(It may not move toward the location of the request to serve it, but before serving it, change to a different request.)

%\adcnote{modifed below. add to below: fraction language, what happens if alg rejects, change $\epsilon$ to 1, change second distance to "time it takes", all distances from a are f}.\ccnote{looks good but added on note and made minor edits.}

\begin{lemma}
\label{lem:last2}
     No non-preemptive deterministic online algorithm can be guaranteed to earn any fraction of the revenue earned by an optimal offline solution in the last $2T/f$ time units.  This is the case whether revenues are uniform or nonuniform.
\end{lemma}
%\ccnote{should we move this lemma to above the general lowerbound?}
\begin{proof}[Proof idea.]  The adversary releases a request in the last two time segments and if the online algorithm chooses not to serve it no other requests will be released.  If the algorithm chooses to serve it, another batch of requests will be released elsewhere that the algorithm cannot serve in time.  Please see Appendix \ref{app:prelim} for details.
\end{proof}

We now present a general lower bound for our problem and  show that \textit{no} non-preemptive deterministic online algorithm (e.g. $\sbp$) can be better than 2-competitive with respect to the revenue earned by the offline optimal schedule (ignoring the last two time segments, due to Lemma \ref{lem:last2}, above).

\begin{theorem} 
\label{thm:genLB} No non-preemptive deterministic online algorithm for OLDARP can be guaranteed to earn more than twice the revenue earned by an optimal offline solution in the first $T-2T/f$ time units.  This is the case whether revenues are uniform or nonuniform.
\end{theorem}

%\ccnote{I think this theorem is for non-preemptive algorithms, in the traditional sense of non-preemptive. Since in case 3 it assumes ON serves $r_1$ or $r_2$ to completion even though $r_3$ and $r_4$ get released 1 time unit after they begin being served.}
\begin{proof}[Proof idea.] The adversary releases requests within the first $T-2T/f$ time segments such that depending on which request(s) the algorithm serves, another set of request(s) with twice as much revenue is released elsewhere that the algorithm cannot serve in time. Please see Appendix \ref{app:prelim} for details.
\end{proof}

\section{ Nonuniform revenues}

In this section we improve the lower and upper bounds for the competitive ratio of the \textsc{segmented best} \textsc{path} algorithm~\cite{atmos17}. In particular, we eliminate the gap between the lower and upper bounds of 4 and 6, respectively, from~\cite{atmos17}, by providing an instance where the lower bound is 5 and a proof for an upper bound of 5.  
Note that throughout this section we assume the revenue earned by $\opt$ in the last two time segments is bounded by some constant.  We must impose this restriction on the $\opt$ revenue of the last two time segments because, as we show in Lemma~\ref{lem:last2}, \textit{no} non-preemptive deterministic online algorithm can be guaranteed to earn any constant fraction of this revenue. 

\subsection{ Lower bound on SBP}\label{sec:lb}
%For the formal proof of the following Theorem, please refer to Appendix \ref{app:lb}.
\begin{figure}
    \begin{center}
 \includegraphics[width=5.5in]{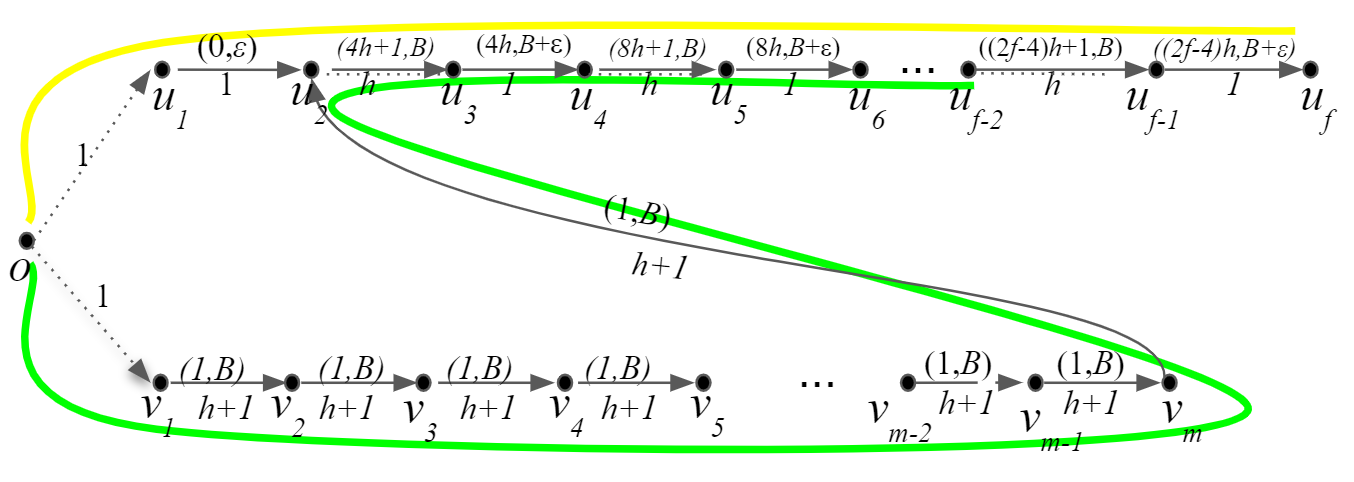}
    \caption{An instance where $\opt$ (whose path is shown in green below) earns $5-4/(f-2)$ times the revenue of $\sbp$ (shown in yellow above).  In this instance, $T=2hf$, and edges that represent requests are shown as solid edges. For each such edge the release time followed by revenue of the corresponding request is shown in parenthesis above the edge. The weight of an edge is shown below the edge. Dashed edges represent empty moves.}
    \label{fig:lb5}
  \end{center}
%  \vspace{-.25in}
\end{figure}
\begin{theorem}\label{thm:lb}
If the revenue earned by $\opt$ in the last two time segments is bounded by some constant, and \textsc{sbp} is $\gamma$-competitive, then $\gamma \ge 5$. 
\end{theorem}
\begin{proof}[Proof idea] For the formal details, please refer to the proof of Theorem \ref{thm:lb} in Appendix \ref{app:lb}.  Consider the instance depicted in Figure \ref{fig:lb5}.   Since $T=2hf$ in this instance, $h$ represents ``half'' the length of one time segment, so only one request of length $h+1$ fits within a single time segment for $\sbp$.  The general idea of the instance is that while $\sbp$ is serving every other request across the top row of requests (since the other half across the top are not released until after $\sbp$ has already passed them by), $\opt$ is serving the entire bottom row in one long chain, then also has time to serve the top row as one long chain.   
\end{proof}

\subsection{ Upper bound on SBP}
\label{sec:ub}
We now show that $\sbp$ is 5-competitive by creating %\adcnote{should we elaborate here about the stipulations?}. 
 a modified, hypothetical $\sbp$ schedule that has additional copies of requests. % First, we lose a factor of two by making additional copies of the requests that $\sbp$ ``incorrectly" serves prior to when they are served by $\opt$. %
First, we note that $\sbp$ loses a factor of 2 due to the fact that it serves requests during only every other time segment.  
%To achieve this, 
Then, we lose another factor of two to cover requests in $\opt$ that overlap between time segments.  Finally, by adding at most one more copy of the requests served by $\sbp$ to make up for requests that $\sbp$ ``incorrectly'' serves prior to when they are served by $\opt$, we end up with 5 copies of $\sbp$ being sufficient for bounding the total revenue of $\opt$. % the set of . %when we create a modified, hypothetical $\sbp$ schedule that has additional copies of the requests already served by $\sbp$ and show that by adding these copies, we exceed the revenue of $\opt$ without adding more than 2 copies of each request served by $\sbp$.  %\adcnote{This last sentence does not make sense to me. If we show that by adding these copies we exceed the revenue of $\opt$ then why not just stop there?}\ccnote{you mean why add the 5th copy?  that is for the requests that overlap across two time segments right?} We lose another factor of 2 due to the fact that $\sbp$ serves requests during only every other time segment, and
%Finally, by adding at most one more copy of the requests served by $\sbp$ to cover requests in $\opt$ that overlap between time segments, we end up with 5 copies of $\sbp$ being sufficient for bounding the revenue of $\opt$.  %\adcnote{modified:}\ccnote{sounds good}
Note that while this proof uses some of the techniques of the proof of the 6-competitive upper bound in~\cite{atmos17}, it reduces the competitive ratio from 6 to 5 by cleverly extracting the set of requests that $\sbp$ serves prior to $\opt$ before making the additional copies.

Let $rev(\opt)$ and $rev(\sbp)$ denote the total revenue earned by $\opt$ and $\sbp$ over all time segments $t_j$ from $j=1 \ldots f$.

\begin{theorem}
\label{onvsopt}
If the revenue earned by $\opt$ in the last two time segments is bounded by some constant, then \textsc{sbp} is 5-competitive, i.e., if $rev(\opt({t_f})) + rev(\opt({t_{f-1}})) \leq c$ for some constant $c$, then $\sum_{j=1}^f rev(\opt({t_j})) \le 5 \sum_{j=1}^f rev(\sbp({t_j}))  + c .$
%\adcnote{Should we move this eqn below into the Theorem statement to save some space?}\ccnote{yes it should be inside the thm statement}
%\begin{equation}
    %rev(S^*({t_f})) + rev(S^*({t_{f-1}}))
%\end{equation}
% Note that we impose this restriction on  $rev(\opt({t_f})) + rev(\opt({t_{f-1}}))$ because, as we show Lemma~\ref{no2tf}, \textit{no} non-preemptive deterministic online algorithm can earn this revenue.  
We note that another interpretation of this result is that under a resource augmentation model where $\sbp$ has two more time segments available than $\opt$, $\sbp$ is 5-competitive.
\end{theorem}

\begin{proof}
We analyze the revenue earned by $\sbp$ by considering the time segments in pairs (recall that the length of a time segment is $T/f$ for some $1 < f <T$). We refer to each pair of consecutive time segments as a time window, so if there are $f$ time segments, there are $\lceil f/2 \rceil$ time windows.  Note that the last time window may have only one time segment.

For notational convenience we consider a modified version of the $\sbp$ schedule, that we refer to as $\sbp'$, which serves exactly the same set of requests as $\sbp$, but does so one time window earlier. Specifically, if $\sbp$ serves a set of requests during time window $i \ge 2$, $\sbp'$ serves this set during time window $i-1$ (so $\sbp'$ ignores the set served by $\sbp$ in window 1). We note that the schedule of requests served by $\sbp'$ may be infeasible, and that it will earn at most the amount of revenue earned by $\sbp$.

 Let $B_i$ denote the set of requests served by  $\opt$ in window $i$ that $\sbp^{\prime}$ already served \emph{before} in some window $j < i$.  
 And let $B$ be the set of all requests that have already been served by $\sbp^{\prime}$ in a previous window by the time they are served in the $\opt$ schedule. Formally, $B = \bigcup_{i=2}^{\lceil f/2 \rceil} B_i$. %(note this includes the requests served by $\opt$ from both time segments in a time window, not only the time segment of greater revenue).
 Consider a schedule $\overline{\opt}$ that contains all of the requests in the $\opt$ schedule minus the requests in $B$. So $\overline{\opt}$ earns total revenue $rev(\opt) - rev(B)$, where $rev(B)$ denotes the total revenue of the set $B$.  

Let $\overline{\opt}(t_j)$ denote the set of requests served by $\overline{\opt}$ in time segment $t_j$. 
Let $\overline{\opt}_i$ denote the set of requests served by $\overline{\opt}$ in the time segment of window $i$ with greater revenue, i.e. $\overline{\opt}_i = \arg\max\{rev(\overline{\opt}({t_{2i-1}})),rev(\overline{\opt}({t_{2i}}))\}$. Note this set may include a request that was started in the prior time segment, as long as it was completed in the time segment of $\overline{\opt}_i$.
Let $rev(\overline{\opt_i})$ denote the revenue earned in $\overline{\opt_i}$.

%\adcnote{added2}
Let $\sbp'_i$ denote the set of requests served by $\sbp'$ in window $i$ and let $rev(\sbp'_i)$ denote the revenue earned by $\sbp'_i$. 
Let $H$ denote the chronologically ordered set of time windows $w$ where $rev(\overline{\opt}_w) > rev(\sbp'_w)$, and let $h_j$ denote the $j$th time window in $H$. We refer to each window of $H$ as a window with a ``hole,'' in reference to the fact that $\sbp'$ does not earn as much revenue as $\overline{\opt}$ in these windows. 
In each window $h_j$ there is some amount of revenue that $\overline{\opt}$ earns that $\sbp'$ does not. In particular, there must be a set of requests that $\overline{\opt}$ serves in window $h_j$ that $\sbp'$ does not serve in $h_j$. Note that this set must be available for $\sbp'$ in $h_j$ since $\overline{\opt}$ does not include the set $B$. 

Let $\overline{\opt}_{h_j} = {A}_j \cup {C}^*_j$, where ${A}_j$ is the subset of requests served by both $\overline{\opt}$ and ${\sbp'}$ in ${h}_j$ and $C^*_j$ is the subset of $\overline{\opt}$ requests available for ${\sbp'}$ to serve in $h_j$ but ${\sbp'}$ chooses not to serve. 
Let us refer to the set of requests served by $\sbp'$ in $h_j$ as $\sbp'_{h_j} =A_j \cup C_j$ for some set of requests $C_j$.  Note that if $\overline{\opt}_{h_j} = A_j \cup C^*_j$ can be executed within a single time segment, then $rev(C_j) \ge rev(C^*_j)$  by the greediness of $\sbp'$.  However, since $h_j$ is a hole we know that the set $\overline{\opt}_{h_j}$ cannot be served within one time segment.

%%%%%%%%%%%%%%%%%%%%%%%%%%%%%%%%%%%%%%%%%%%%%%

Our plan is to build an infeasible schedule $\overline{\sbp}$ that will be similar to $\sbp'$ but contain additional ``copies'' of some requests such that no windows of $\overline{\sbp}$ contain holes.  %---We will make no more than $3$ copies of each request in $\sbp'$, so ultimately $\overline{\sbp}$ will have revenue at most $3$ times that of $\sbp'$.  ----
We first initialize $\overline{\sbp}$ to have the same schedule of requests as $\sbp'$. 
We then add additional requests to $h_j$  for each $j = 1 \ldots |H|$, based on $\overline{\opt}_{h_j}$.

Consider one such window with a hole $h_j$, and let $k$ be the index of the time segment corresponding to $\overline{\opt}_{h_j}$. 
We know $\overline{\opt}$ must have begun serving a request of $\overline{\opt}_{h_j}$ in time segment $t_{k-1}$ and completed this request in time segment $t_k$. Let us use $r^*$ to denote this request that ``straddles'' the two time segments.

After the initialization of $\overline{\sbp}=\sbp'$, recall that the set of requests served by $\overline{\sbp}$ in $h_j$ is ${\overline{\sbp}}_{h_j} =A_j \cup C_j$ for some set of requests $C_j$. 
We add to $\overline{\sbp}$ a copy of a set of requests. 
There are two sub-cases depending on whether $r^*\in C^*_j$ or not.  

Case $r^* \in C^*_j$.  In this case, by the greediness of $\sbp$, and the fact that both $r^*$ alone and $C^*_j\setminus{\{r^*\}}$ can separately be completed within a single time segment, we have: $rev(C_j) \ge \max\{rev(r^*),rev(C^*_j\setminus \{r^*\})\} \ge \frac{1}{2}rev(C^*_j).$
%in addition to the copy of $B_j$, we also
We then add a copy of the set $C_j$ to the $\overline{\sbp}$ schedule, so there are two copies of $C_j$ in $h_j$.  
Note that for $\overline{\sbp}$, $h_j$ will no longer be a hole since: 
%$$rev(\overline{\opt}_{h_j}) = rev({A}_j) + rev({B}_j) + rev({C}^*_j) \leq rev(A_j) + rev(B_j) + 2\cdot rev(C_j) = rev(\overline{S}_{h_j}).$$
%$$rev(\overline{\opt}_{h_j}) = rev({A}_j)  + rev({C}^*_j) \leq rev(A_j) + 2\cdot rev(C_j) = rev(\overline{S}_{h_j}).$$
$rev(\overline{\opt}_{h_j}) = rev({A}_j)  + rev({C}^*_j) \leq rev(A_j) + 2\cdot rev(C_j) = rev(\overline{\sbp}_{h_j}).$
%where $\overline{S}_{h_j}$ is the set of requests served by $\overline{\sbp}$ in $h_j$. \\

Case $r^* \notin C^*_j$.  In this case %either $r^* \in A$ or $r^* \in B$.  If $r^* \in A$ then 
$C^*_j$ can be served within one time segment but $\sbp'$ chooses to serve $A_j \cup C_j$ instead. So we have $rev(A_j) + rev(C_j) \ge rev(C^*_j)$, therefore we know either $rev(A_j) \ge \frac{1}{2}rev(C^*_j)$ or $rev(C_j) \ge \frac{1}{2}rev(C^*_j)$.  In the latter case, we can do as we did in the first case above and add a copy of the set  $C_j$ to the $\overline{\sbp}$ schedule in window $h_j$, to get $rev(\overline{\opt}_{h_j})  \leq rev(\overline{\sbp}_{h_j})$, as above.  In the former case, we instead add a copy of $A_j$ to the $\overline{\sbp}$ schedule in window $h_j$. Then again,  for $\overline{\sbp}$, $h_j$ will no longer be a hole, since this time:
%$$rev(\overline{\opt}_{h_j}) = rev({A}_j) + rev({B}_j) + rev({C}^*_j) \leq 2\cdot rev(A_j) + rev(B_j) + rev(C_j) = rev(\overline{\sbp}_{h_j}).$$
$rev(\overline{\opt}_{h_j}) = rev({A}_j) + rev({C}^*_j) \leq 2\cdot rev(A_j) + rev(C_j) = rev(\overline{\sbp}_{h_j}).$
%\end{enumerate}

%Note that since all requests that have been previously served by \sbp’ have been removed in the first step (above), no “new” B’s will be created.
Note that for all windows $w \notin H$  
that are not holes, we already have $rev(\sbpbar_w) \ge rev(\optbarprime_w)$. So we have
\begin{eqnarray}
\label{vsoptbarprime}
 \sum_{i=1}^{\lceil f/2 \rceil -1} rev(\optbarprime_i) \le \sum_{i=1}^{\lceil f/2 \rceil -1} rev(\sbpbar_i ) \le 2\sum_{i=1}^{\lceil f/2 \rceil -1} rev( \sbpprime_i). \end{eqnarray}
where the second inequality is because $\sbpbar$ contains no more than two instances of every request in $\sbpprime$. %, so  we have
%\begin{eqnarray}
%\label{vssbpprime}
%\sum_{i=1}^{\lceil f/2 \rceil -1} rev(\sbpbar_i) \le 2\sum_{i=1}^{\lceil f/2 \rceil -1} rev( \sbpprime_i)
%\end{eqnarray}
Combining (\ref{vsoptbarprime}) % and (\ref{vssbpprime})
with the fact that $\sbpprime$ earns at most what $\sbp$ does yields
%\begin{eqnarray}
%\sum_{i=1}^{\lceil f/2 \rceil -1} rev(\optbarprime_i) \le 2\sum_{i=1}^{\lceil f/2 \rceil -1} rev(\sbpprime_i)
%\end{eqnarray}
%Since $\sbpprime$ earns at most the revenue of $\sbp$ we have
\begin{equation}
\label{sum1}
\sum_{i=1}^{\lceil f/2 \rceil} rev(\optbarprime_i) \le 2 \sum_{i=1}^{\lceil f/2 \rceil} rev(\sbp_i) + rev(\overline{\opt}(t_{f-1})) + rev(\overline{\opt}(t_f)) .
\end{equation}
Since $\sbp$ serves in only one of two time segments per window, we have
%\begin{eqnarray}
%\label{sum2}
$\sum_{i=1}^{\lceil f/2 \rceil} rev(\sbp_i) = \sum_{j=1}^f  rev(\sbp(t_j)).$
%\end{eqnarray}
Hence, by the definition of $\optbarprime$, and by (\ref{sum1}) %and (\ref{sum2}), 
we can say
\begin{eqnarray}
\label{sum3}
\sum_{j=1}^f rev(\overline{\opt}(t_j)) \le  2\sum_{i=1}^{\lceil f/2 \rceil} rev( \optbarprime_i) \le  4\sum_{j=1}^f rev(\sbp(t_j)) + rev(\overline{\opt}(t_{f-1})) + rev(\overline{\opt}(t_f)).
\end{eqnarray}
%
%Equations (\ref{sum1})-(\ref{sum3}) imply			
%\begin{eqnarray}
%\label{sum4}
%\sum_{j=1}^f rev(\overline{\opt}(t_j)) \le  4\sum_{j=1}^f rev(\sbp(t_j)) + rev(\overline{\opt}(t_{f-1})) + rev(\overline{\opt}(t_f)) 
%\end{eqnarray}
Now we must add in any request in $B$, such that $\opt$ serves the request in a time window after $\sbpprime$ serves that request. By definition of $B$ (as the set of all requests that have been served by $\sbpprime$ in a previous window) $B$ may contain at most the  same set of requests served by $\sbpprime$. Therefore $rev(B) \le rev(\sbpprime)$, so
%\begin{equation}
%\label{end1}
$rev(B) \le rev(\sbp).$
%\end{equation}
By the definition of $\opt$, $\opt = \overline{\opt} + B$, so			
\begin{equation}
\label{end2}
\sum_{j=1}^f rev(\opt(t_j)) =  rev(B) + \sum_{j=1}^f rev(\overline{\opt}(t_j))	
\end{equation}
And by combining (\ref{sum3})-(\ref{end2}) with the fact that $rev(B) \le rev(\sbp)$, we have
\begin{align*}
\sum_{j=1}^f rev(\opt(t_j)) &\le& \sum_{j=1}^f rev(\sbp(t_j)) + 4\sum_{j=1}^f rev(\sbp(t_j)) + rev(\overline{\opt}(t_{f-1})) + rev(\overline{\opt}(t_f))	\\	%\end{eqnarray*}				So
%\begin{eqnarray*}
%\sum_{j=1}^f rev(\opt(t_j)) 
&\le& 5 \sum_{j=1}^f rev(\sbp(t_j)) + rev(\overline{\opt}(t_{f-1})) + rev(\overline{\opt}(t_f)).
\end{align*}

\end{proof}

\section{ Uniform revenues}
\label{sec:uniformrev}

We now consider the setting where revenues are uniform among all requests, so the goal is to maximize the total number of requests served. This variant is useful for settings where all requests have equal priorities, for example for not-for-profit services that provide transportation to elderly and disabled passengers. %\ccnote{added} 
%The proof strategy %here is also more surgical, while still having quite an elegant and clean feel.  The idea 
%is to track categories of requests served on the $\opt$ side versus the $\sbp$ side, and show that they all ultimately map to one another.
%\adcnote{I don't understand this comment - are we really mapping like this? Maybe a little but I don't think that is the main proof strategy.}\ccnote{yeah i guess you are right, this was just in reference to lemma 5... but maybe the rest of it is like the previous proof?} \adcnote{yes, similar to the previous proof so we still need new  "summary" text about this proof}. 
%Also, I found this language from an earlier version of the paper that we can use if we are still looking for some:  
The proof strategy is to carefully consider  the requests served by $\sbp$ in each window and track how they differ from that of $\opt$.  The final result is achieved through a clever accounting of the differences between the two schedules, and bounding the revenue of the requests that are ``missing'' from $\sbp$. 

We note that the lower bound instance of Theorem \ref{thm:lb} can be modified to become a uniform-revenue instance that has ratio $5-14/f.$
On the other hand, we also show that $\opt$ earns at most 4 times the revenue of $\sbp$ in this setting if we assume the revenue earned by $\opt$ in the last two time segments is bounded by a constant, and allow $\sbp$ an additive bonus of $f$ %\ccnote{should the additive bonus here say $f$ not $f/2$?}. 
Note that when revenues are uniform, \textit{no} non-preemptive deterministic online algorithm can earn $rev(\opt({t_f})) + rev(\opt({t_{f-1}}))$ (see Lemma~\ref{lem:last2}). 
We begin with several definitions and lemmas.% that will aid us in proving Theorem \ref{4-comp-theorem} of Section \ref{sec:uniformrev}.
%\noindent
%\textbf{Part I: Partition the Set}
%As in the proof of Theorem~\ref{onvsopt}, 

As in the proof of Theorem~\ref{onvsopt}, we consider a modified version of the $\sbp$ schedule, that we refer to as $\sbp'$, which serves exactly the same set of requests as $\sbp$, but does so one time window earlier.  For all windows $i=1,2,..., m$, where $m=\lceil f/2\rceil -1$, we let $S'_i$ denote the set of requests served by $\sbp'$ in window $i$ and $S^*_i$ denote
the set of requests served by $\opt$ during the time segment of window $i$ with greater revenue, i.e. $S^*_i =\argmax \{ rev(\opt (t_{2i-1}), rev(\opt(t_{2i})) \}$ where $rev(\opt (t_j))$ denotes the revenue earned by $\opt$ in time segment $t_j$. We define a new set $J_i^*$ as the set of requests served by $\opt$ during the time segment of window $i$ with less revenue, i.e. $J_i^*=\argmin \{ rev(\opt (t_{2i-1}), rev(\opt(t_{2i})) \}$.

Let $S_i^*=A_i \cup X_i^* \cup Y_i^*$, and $S_i'=A_i \cup X_i \cup Y_i$, where:
%\begin{enumerate}
  %  \item 
  (1) $A_i$ is the set of requests that appear in both $S_i^*$ and $S_i'$; (2)
%    \item
$X_i^*$ is the set of requests that appear in $S_w'$ for some $w=1,2,...,i-1$. Note there is only one possible $w$ for each individual request $r\in X_i^*$, because each request can be served only once; 
(3)
%\item
$Y_i^*$ is the set of requests such that no request from $Y_i^*$ appears in $S_w'$ for any $w=1,2,...,i-1,i$; 
(4)
%\item
$X_i$ is the set of requests that appear in $S_w^*$ for some $w=1,2,...,i-1$. Note there is only one possible $w$ for each individual request $r\in X_i$, because each request can be served only once;
(5)
%\item 
$Y_i$ is the set of requests such that no request from $Y_i$ appears in $S_w^*$ for any $w=1,2,...,i-1,i$.
%    \end{enumerate}

Note that elements in $Y_i$ can appear in a previous $J_w^*$ for any $w=1,2,...,i-1, i$ or in a future $S^*_v$ or $J^*_v$ for any $v=i+1,i+2,...,m$, or may not appear in any other sets.
Also note that since each request can be served at most once, we have:
% the intersection of each pair formed by elements of the following sets is empty: 
%\begin{enumerate}
 %   \item
 $A_1 \cap X_1^* \cap Y_1^* \cap A_2 \cap X_2^* \cap Y_2^*\cap...\cap A_{m}\cap X_{m}^*\cap Y_{m}^* = \emptyset$ 
  %  \item
  and
  $A_1\cap X_1\cap Y_1\cap A_2\cap X_2\cap Y_2\cap ...\cap A_{m}\cap X_{m}\cap Y_{m} = \emptyset$.
%\end{enumerate}

Given the above definitions, we have the following lemma whose proof has been deferred to Appendix \ref{app:uni-rev}.  It states that at any given time window, the cumulative requests of $\opt$ that were earlier served by $\sbp$ are no more than the number that have been served by $\sbp$ but not yet by $\opt$.

\begin{lemma}
    \label{lemma-a-4}
    $|X_1^*|+|X_2^*|+...+|X_i^*| \leqslant |Y_1|+|Y_2|+...+|Y_{i-1}| + |Y_i|$ for all $i=1,2,...,m$
\end{lemma}

We are now ready to prove our main theorem of this section.
%We now restate Theorem \ref{4-comp-theorem} for the reader's convenience before presenting its proof.
\begin{theorem}
\label{4-comp-theorem}

    If the revenue earned by $\opt$ in the last two time segments is bounded by some constant, $c$, then $\sbp$ earns at least 1/4 the revenue of the optimal solution, minus an additive term linear in $f$, where $T/f$ is the length of one time segment. (So if $f$ is bounded by some constant, then $\sbp$ is 4-competitive).
%If the number of time segments, $f$, and the revenue earned by $\opt$ in the last two time segments are both bounded by some constant, c, then \textsc{sbp} is 4-competitive, i.e., if $rev(\opt({t_f})) + rev(\opt({t_{f-1}})) \leq c$ and $\lceil f/2 \rceil -1=m$ for some constants $c$ and $m$, then
I.e., $    \sum_{j=1}^{f} rev($\opt$(t_j)) \le 4\sum_{j=1}^{f} rev($\sbp$(t_j))+2\lceil f/2 \rceil+c$. % Note that \textit{no} non-preemptive deterministic online algorithm can earn $rev(\opt({t_f})) + rev(\opt({t_{f-1}}))$ (see Lemma~\ref{lem:last2}).%\ccnote{I don't know if we can use Lemma 2 here b/c it uses nonuniform revenues...} \adcnote{Ok, I have something for this too - just have adversary release k requests, so total revenue is k}
\end{theorem}
%\ccnote{Should we define $m$ and $c$ in the theorem statement here? }
%\noindent
%\textbf{Part II: Inequalities Derived from Greediness}
\begin{proof}
    Note that since revenues are uniform, the revenue of a request-set $U$ is equal to the size of the set $U$, i.e., $rev(U)=|U|$. 
    Consider each window $i$ where $rev(S_i^*)>rev(S_i')$. Note that the set $S_i^*$ may not fit within a single time segment. We consider two cases based on $S_i^*$. 
    \begin{enumerate}
    \item The set $S_i^*$ can be served within one time segment. Note that within $S_i^*=A_i \cup X_i^* \cup Y_i^*$, $X_i^*$ is not available for $\sbp'$ to serve because $\sbp'$ has served the requests in $X_i^*$ prior to window $i$. Among requests that are available to $\sbp'$, $\sbp'$ greedily chooses to serve the maximum revenue set that can be served within one time segment. Therefore, we have
   
    %\begin{eqnarray*}
        $rev(X_i)+rev(Y_i) \geqslant rev(Y_i^*).$
    %    \text{   (by greediness)}
%    \end{eqnarray*}
    Since revenues are uniform, we also have
 %   \begin{eqnarray*}
  $      |X_i| + |Y_i| \geqslant |Y_i^*|.$
    %    \text{   (uniform revenue)}
    %\end{eqnarray*}
    
    If this is not the case, then $\sbp'$ would have chosen  to serve $Y_i^*$ instead of $X_i \cup Y_i$ since it is feasible for $\sbp'$ to do so because the entire $S_i^*$ can be served within one time segment.

    \item The set $S_i^*$ cannot be served within one time segment. This means there must be one request in $S_i^*$ that $\opt$ started serving in the previous time segment. We refer to this straddling request as $r^*$. There are three sub-cases based on where $r^*$ appears. %(a) $r^* \in Y_i^*$; (b) $r^* \in X_i^*$; (c) $r^* \in A_i$. 
    \begin{enumerate}
    \item  If $r^* \in Y_i^*$, then due to the greediness of $\sbp'$, we know that
    \begin{eqnarray}
        \label{eq-2a-1}
        rev(X_i) + rev(Y_i) \ge rev(r^*)
    \end{eqnarray}
    since otherwise $\sbp'$ would have chosen to serve $r^*$. We also know 
    \begin{eqnarray}
        \label{eq-2a-2}
        rev(X_i) + rev(Y_i) \ge rev(Y_i^* \backslash \{r^*\})
    \end{eqnarray}
    since otherwise $\sbp'$ would have chosen to serve $Y_i^*\backslash \{ r^* \}$. 
    
    From (\ref{eq-2a-1}), we have $|X_i| + |Y_i| \geqslant 1$ and from (\ref{eq-2a-2}), we have $|X_i| + |Y_i| \geqslant |Y_i^*| - 1.$
    
    \item If $r^* \in X_i^*$, then $r^*$ is not available to $\sbp'$ and only $A_i$, $X_i$, $Y_i$, and $Y_i^*$ are available to $\sbp'$. Therefore we know that
    %$$rev(A_i)+rev(X_i) +rev(Y_i)\geqslant rev(A_i)+rev(Y_i^*)$$
    $rev(X_i)+rev(Y_i)\geqslant rev(Y_i^*)$
    since otherwise, by its greediness,  $\sbp'$ would have chosen to serve $A_i$ and $Y_i^*$ instead of $A_i$, $X_i$ and $Y_i$, because $A_i$ and $Y_i^*$ can be served within one time segment.
    % considering that $r^*\in X_i^*$, which ensures $A_i$ and $Y_i^*$ to fit within one time segment. 
    Therefore, we have
    $|X_i| + |Y_i| \geqslant |Y_i^*|.$
    \item $r^* \in A_i$. Then $r^*$ is served by both $\opt$  and $\sbp'$. %Although we are not sure whether $A_i\cup Y_i^*$ can be served within one time segment, we are certain that 
    We know that $A_i\cup Y_i^* \backslash \{r^*\}$ can be served within one time segment since $r^*$ is the only request that causes $S_i^*$ to straddle between two time segments.  %, and removing $r^*$ will cause the straddling to disappear.
    Again by the greediness of $\sbp'$, we have $rev(A_i)+rev(X_i)+rev(Y_i)\geqslant rev(A_i)+rev(Y_i^*)-rev(r^*)$
    which means
    $rev(X_i)+rev(Y_i)\geqslant rev(Y_i^*)-rev(r^*)$ and
    $|X_i|+|Y_i|\geqslant |Y_i^*|-1$.
    %Note that if this inequality were note true, then $\sbp'$ would have chosen to serve $A_i\cup Y_i^* \backslash \{r^*\}$ instead. 
    \end{enumerate}
    \end{enumerate}
    
    Therefore, for all cases, for window $i$, we have $
        |X_i|+|Y_i|\geqslant|Y_i^*|-1$, which means
   % \begin{equation}
    %    \label{partii}
     $   |Y_i^*|-|X_i| \leqslant 1+|Y_i|$,
    %\end{equation}
    %\textbf{Part III: Filling the Hole}
    and with  $m=\lceil f/2 \rceil -1$, 
    \begin{equation}
        \label{hole-eq-3}
        \sum_{i=1}^{m}(|Y_i^*|-|X_i|) \leqslant m+ \sum_{i=1}^{m}|Y_i|.
    \end{equation}
    %where $m=\lceil f/2 \rceil -1$ as previously defined. 
    
    Now we will build an infeasible schedule $\overline{\sbp}$ that will be similar to $\sbp'$ but contain additional ``copies'' of some requests such that no windows of $\overline{\sbp}$ contain holes, i.e. such that
%    \begin{equation*}
 $       rev(\overline{\sbp}) \geqslant \sum_{i=1}^{m}rev(S_i^*). 
  $%  \end{equation*}

    %We define $\Phi PT=\cup_{i=1}^{m} S_i^*$ and observe that
    We define a modified $\opt$ schedule which we refer to as $\opt'$ such that $\opt'=\cup_{i=1}^{m} S_i^*$ and observe that
   % \begin{equation*}
    $    rev(\opt')=\sum_{i=1}^{m} |A_i| + \sum_{i=1}^{m} |X_i^*| + \sum_{i=1}^{m} |Y_i^*|
    $%\end{equation*}
   , while
    %\begin{equation*}
     $   rev(\sbp')=\sum_{i=1}^{m} |A_i| + \sum_{i=1}^{m} |X_i| + \sum_{i=1}^{m} |Y_i|.
$%    \end{equation*}
    %Now we compare $rev(\opt')$ and $rev(\sbp')$ by finding their difference:
   
    By Lemma \ref{lemma-a-4} and equation (\ref{hole-eq-3}), we can say
    \begin{eqnarray}
    %    \label{hole-eq-4}
        rev(\opt')-rev(\sbp')&=& \sum_{i=1}^{m} |Y_i^*|- \sum_{i=1}^{m} |X_i| + \sum_{i=1}^{m} |X_i^*| - \sum_{i=1}^{m} |Y_i|\\
    %\end{eqnarray}
%    \begin{eqnarray}
 %       \label{hole-eq-1}
  %      \sum_{i=1}^{m}|X_i^*| \leqslant \sum_{i=1}^{m}|Y_i|
   % \end{eqnarray}
%    From (\ref{hole-eq-1}) and (\ref{hole-eq-4}), we have
  %  \begin{eqnarray}
        \label{hole-eq-5}
    %  rev(\opt')-rev(\sbp')
    &\leqslant& \sum_{i=1}^{m} |Y_i^*|- \sum_{i=1}^{m} |X_i|\leqslant m+ \sum_{i=1}^{m} |Y_i|.
    \end{eqnarray}
   % From (\ref{hole-eq-3}) and (\ref{hole-eq-5}), we have
    %\begin{equation}
     %   \label{hole-eq-6}
        %rev(\opt')-rev(\sbp')\leqslant m+ \sum_{i=1}^{m} |Y_i|.
    %\end{equation}
Inequality (\ref{hole-eq-5})
tells us that to form an $\overline{\sbp}$ whose revenue is at least that of $\opt'$, we must “compensate” $\sbp'$ by adding to it at most copies of all requests in the set $Y_i$ for all $i=1,2,...,m$, plus $m$ “dummy requests.” In other words, 
    \begin{equation}
        \label{hole-eq-7}
        rev(\overline{\sbp})=rev(\sbp')+m+\sum_{i=1}^{m} |Y_i| \geqslant rev(\opt').
    \end{equation}
    %(\ref{hole-eq-7}) is a direct result from (\ref{hole-eq-6}). 
    We know the total revenue of all $Y_i$ can not exceed the total revenue of $\sbp'$, hence we have
    \begin{equation}
        \label{hole-eq-8}
        rev(\overline{\sbp})=rev(\sbp')+m+\sum_{i=1}^{m} |Y_i| \leqslant 2rev(\sbp')+m.
    \end{equation}
    Combining (\ref{hole-eq-7}) and (\ref{hole-eq-8}), we get 
%    \begin{eqnarray*}
 $       rev(\opt') \leqslant 2rev(\sbp')+m$, which means
%    \end{eqnarray*}
    \begin{equation}
        \label{hole-eq-9}
        \sum_{i=1}^{m} rev(S_i^*)\leqslant 2\sum_{i=1}^{m} rev(S_i') + m.
    \end{equation}
    Recall that $S_i^*$ is the set of requests served by $\opt$ during the time segment of window $i$ with greater revenue. In other words,
   % \begin{equation}
    %    \label{hole-eq-10}
$        \sum_{j=1}^{2m} rev(S^*(t_j)) \leqslant 2\sum_{i=1}^{m} rev(S_i^*), $
%    \end{equation}
 which, combined with  (\ref{hole-eq-9}), gives us% and (\ref{hole-eq-10}), we get
    \begin{equation}
        \label{hole-eq-11}
        \sum_{j=1}^{2m} rev(S^*(t_j)) \leqslant 4\sum_{i=1}^{m} rev(S_i') + 2m.
    \end{equation}
        We assumed that the total revenue of requests served in the last two time segments by $\opt$ is bounded by $c$. From (\ref{hole-eq-11}), we get
    \begin{equation}
        %\label{hole-eq-temp-0}
           \label{hole-eq-12}
        \sum_{j=1}^{f} rev(S^*(t_j)) \leqslant\sum_{j=1}^{2m} rev(S^*(t_j)) + rev(S^*(t_{f-1})) + rev(S^*(t_f)) \leqslant 4\sum_{i=1}^{m} rev(S_i') + 2m + c.
    \end{equation}
    %Consider the left-hand-side of (\ref{hole-eq-temp-0}) which we will refer to as $LHS$. If $f$, the number of time segments, is even, then $LHS=\sum_{j=1}^{f} rev(S^*(t_j))$. If the number of time segments is odd, then $LHS\geqslant \sum_{j=1}^{f} rev(S^*(t_j))$ since the revenue earned by $\opt$ during the $(f-1)^{th}$ time segment is added twice. Therefore:
    %Also note that
    %\begin{equation}
   %     \label{hole-eq-temp-1}
     %   \sum_{j=1}^{f} rev(S^*(t_j)) \leqslant \sum_{j=1}^{2m} rev(S^*(t_j)) + rev(S^*(t_{f-1})) + rev(S^*(t_f))
   % \end{equation}
    %since the revenue earned by $\opt$ during the $(f-1)^{th}$ time segment is added twice in the event that $f$ is odd.
   % Combining (\ref{hole-eq-temp-0}) and (\ref{hole-eq-temp-1}), we get
%    \begin{eqnarray}
%        \label{hole-eq-12}
%        \sum_{j=1}^{f} rev(S^*(t_j)) \leqslant 4\sum_{i=1}^{m} rev(S_i') + 2m + c
%    \end{eqnarray}
    We also know that the total revenue of requests served by $\sbp'$ during the first $m$ windows is less than or equal to the total revenue of $\sbp$. Therefore, from (\ref{hole-eq-12}), we have
%    \begin{equation}
 %       \nonumber
   $     \sum_{j=1}^{f} rev(S^*(t_j)) \leqslant 4\sum_{j=1}^{f} rev(S(t_j)) + 2m + c.$
  %  \end{equation}
%    as needed.
%    Therefore we have
%    \begin{eqnarray*}
 %       \sum_{j=1}^{f} rev(\opt(t_j)) \leqslant 4\sum_{j=1}^{f} rev(\sbp(t_j))+2m+c
  %  \end{eqnarray*}
\end{proof}

%Due to space limitations, we have regretfully deferred the proof of the following theorem to 
%Appendix. Please refer to 
%Appendix \ref{app:uni-rev}.
%for the proof of the following theorem.

%\noindent
%\textbf{Part II: Inequalities Derived from Greediness}
%\begin{proof}
%Please see 
%Appendix \ref{app:uni-rev}.
%\end{proof}

\section{ Bipartite graphs}
\label{sec:bipartite}
In this section, we consider ROLDARP for complete bipartite graphs $G=(V=V_1\cup V_2,E)$, %specifically where if $V_1$ and $V_2$ denote the two sets of nodes, 
where only nodes in $V_1$ maybe be source nodes and only nodes in $V_2$ may be destination nodes.  
%every source is in $V_1$ and every node is designated as either a source node (in $V_1$) or a destination node (in $V_2$) destination is in $V_2$. 
One node is designated as the
origin and there is an edge from this node to every node in $V_1$ (so the origin is a node in $V_2$). Due strictly to space limitations, most proofs of theorems in this section are deferred to Appendix \ref{app:bipartite}.

We refer to this problem as ROLDARP-B and the offline version as RDARP-B. We first show that if edge weights of the bipartite graph are not bounded by a minimum value, then the offline version of ROLDARP on general graphs, which we refer to as RDARP, reduces to RDARP-B. Since RDARP has been show in \cite{atmos17,twochain} to be NP-hard (even if revenues are uniform), this means RDARP-B is NP-hard as well.
%We therefore impose a minimum edge weight of $kT/f$ for some constant $k$ such that $0<k\leqslant 1$ (the maximum edge weight is still $T/f$ as before). %We prove that %a modified version of 
	%Segmented Best Path ($\sbp$) %, which we refer to as $\sbp-b$ (for bipartite graphs), 
	%is $\lceil 1/k \rceil$-competitive, regardless of whether the revenues of requests are uniform.

\begin{theorem}\label{thm:roldarpb}
    The problem RDARP is poly-time reducible to RDARP-B. Also, RDARP with uniform revenues is poly-time reducible to RDARP-B with uniform revenues. %\adcnote{do we need to show that ROLDARP with uniform revenues is NP-hard?}\ccnote{well i guess it is the offline version of the problem that is hard right? and for the offline version in the other paper with barbara we showed that it is hard even when everything is uniform... making me wonder if when we say poly-time reducible here in this theorem if we are actually talking about the offline problem not the online... we may need to fix this...} \adcnote{I think we are talking about the online problem}
%    \ccnote{after reading the proof i think we are talking about the offline problems here. so we will need to edit this section...}
\end{theorem}

\begin{proof}[Proof idea.]  
The idea of the reduction is to split each node into two nodes connected by an edge in the bipartite graph with a distance of $\epsilon$. Then we turn each edge in the original graph into two edges in the bipartite graph.  Please see Appendix for \ref{app:bipartite} details.
    %In the bipartite graph we can have requests only from a node on the source side to a node on the destination side (and not vice-versa), so we make edges $(u, v’)$ and $(v’, u)$ have weight $w(u, v)$. Finally, we make edges $(u, u’)$ have weight 0 (so we can serve a chain like $a-b-c$ as $(a, b’)$, $(b’, b)$, $(b, c’))$.
\end{proof}

%\RestyleAlgo{boxruled}
%\begin{algorithm}\caption{Algorithm \textsc{Segmented Best Path-Bipartite  (\sbp-B)}. Input is complete graph $G$ with time limit $T$ and maximum edge weight $T/f$.}
%\label{sbp-b}
%\begin{algorithmic} [1]
%\STATE Let $t_1, t_2, \ldots t_{f}$ denote the time segments ending at times $T/f, 2T/f, \ldots, T$, \\ respectively.
%\IF {$f$ is odd}
%\STATE At $t_1$, do nothing.
% \STATE At the start of every $t_i$ for even $i \ge 2$ find the \textit{max-revenue-request-set} and \\ move to the source location of the first request in this set. Denote this request \\ set as $R$. If no unserved request sets exist, do nothing.
% \STATE At the start of every $t_i$ for odd $i \ge 3$, serve request set $R$ (if it exists) from the previous step.
% \ENDIF
%\IF {$f$ is even}
%\STATE At the start of every $t_i$ for odd $i \ge 1$ find the \textit{max-revenue-request-set} and move \\ to the source location of the first request in this set. Denote this request set as \\ $R$. If no unserved request sets exist, do nothing.
%\STATE At the start of every $t_i$ for even $i \ge 2$, serve request set $R$ (if it exists) from the previous step.
%\ENDIF
%\end{algorithmic}
%\end{algorithm}

\subsection{ Uniform revenue bipartite}
We show that for bipartite graph instances, if revenues are uniform, we can guarantee that $\sbp$ earns a fraction of $\opt$ equal to the ratio between the  minimum and maximum edge-length.
\begin{theorem}
    \label{thm:r-uniform-comp}
    For any instance of ROLDARP-B where the revenues are uniform for all requests, if edge weights are upper and lower bounded by $T/f$ and $kT/f$, respectively, for some constant $0<k\leqslant 1$, then $$rev(\opt)\leqslant \lceil 1/k \rceil \cdot rev(\sbp) + \lceil 1/k \rceil.$$
\end{theorem}
%\ccnote{Do we need a complete bipartite graph for this theorem? Above in the intro of this section we don't say anything about complete bipartite graphs... and should we use the term ROLDARP-B somewhere in this theorem statement?}
%\adcnote{I think we need a proof idea here and for the next theorem.} 

\begin{proof}[Proof idea.]
The proof idea is akin to that of Theorem \ref{thm:r-nonuniform-comp} below. 
%If SBP serves at least 1 request per window, then the result is immediate. If SBP serves no requests in some window, then we break down the revenue of SBP and OPT into three partitions based on the location of the last such window: before the last empty window, during, and after. 
Please see Appendix~\ref{app:bipartite} for the details.
\end{proof}

\subsection{ Nonuniform revenue bipartite}\label{sec:bipartite-nonuniform}

In this section we show that even if revenues are nonuniform, we can still guarantee that $\sbp$ earns a fraction of $\opt$ equal to the ratio between the minimum and maximum edge-length, minus the revenue earned by $\opt$ in the last window.  Recall that we refer to each pair of consecutive time segments as a time window. %, so if there are $f$ time segments, there are $\lceil f/2 \rceil$ time windows.  Note that the last time window may have only one time segment.%(i.e. last two time segments).  
Note that \textit{no} non-preemptive deterministic online algorithm can be competitive with any fraction of the revenue earned by an optimal offline solution in the last $2T/f$ time units (i.e. Lemma~\ref{lem:last2} also holds for ROLDARP-B with nonuniform revenues).%, see Lemma \ref{lem:last2b} in Appendix \ref{app:prelim}).% \adcnote{We need to confirm this last sentence. Also, I am a little unsure about using windows vs time segments here. If the number of time segments is odd, then is it still correct to say the "last window" in the theorem statement below?}. \ccnote{I am also unsure.  I think to be safe we should just remove the claim here about lemma 1 applying and also change it to last two time segments rather than window.}

\begin{theorem}
    \label{thm:r-nonuniform-comp}
    %In a complete bipartite graph 
    For any instance of ROLDARP-B where the revenues of requests are nonuniform, if edge weights are upper and lower bounded by $T/f$ and $kT/f$, respectively, for some constant $0<k\leqslant 1$, and if the revenue earned by $\opt$ in the last %$2T/f$ units of time 
    time window
    is bounded by some constant $c$, then $$rev(\opt)\leqslant \lceil  1/k \rceil \cdot rev(\sbp) + c$$
\end{theorem}

\begin{proof}[Proof of Theorem \ref{thm:r-nonuniform-comp}]
Again, we refer to each pair of consecutive time segments as a time window. We consider a hypothetical schedule which we refer to as $\sbp'$ that proceeds as follows.  
%    \begin{enumerate}
     %   \item 
     In the first time window, $\sbp'$ does nothing.
     %   \item 
     In the $i^{th}$ window ($2 \leqslant i\leqslant \lceil f/2 \rceil$),  
        $\sbp'$ serves exactly one request: the maximum revenue request served by $\opt$ in the $(i-1)^{th}$ window. (In Lemmas~\ref{land-of-oz} and \ref{lem:conclude-sbpdb} of Appendix \ref{app:bipartite} we show that the revenue earned by $\sbp'$ is no greater than the revenue earned by $\sbp$.)
    %\end{enumerate}
        Let $Q_i$, $Q_i'$, and $Q_i^*$ denote the sets of requests served by $\sbp$, $\sbp'$, and $\opt$, respectively, in window $i$. There are two cases based on the performance of $\sbp$.
%\begin{enumerate}
%\item

Case 1: $\sbp$ serves at least one request per window.  
%    In this case, we will show $\sbp'$ is $\lceil 1/k \rceil$-competitive, which means $\sbp$ must also be $\lceil 1/k \rceil$-competitive. 
    %Suppose the time limit contains $\mu$ windows.  
    Again let $\mu = \lceil f/2 \rceil$ denote the total number of time windows. 
    Let $r=\lceil 1/k \rceil$. We know from Theorem \ref{thm:r-uniform-comp} that $\opt$ can serve at most $r$ requests per window. 
    We assume without loss of generality that $\opt$ serves exactly $r$ requests per window and
    let $\rho_i = \rho_{i,1}, \rho_{i,2},...,\rho_{i,r}$ denote the $r$ revenues earned by $\opt$ in window $i$.
    Consider the first window of $\opt$ and the second window of $\sbp'$.
    In the first window, $\opt$ earns revenues $\rho_1 = \rho_{1,1}, \rho_{1,2},...,\rho_{1,r}$. In the second window, $\sbp'$  serves the maximum revenue request from $\rho_1$. Therefore, 
    %\begin{eqnarray}
     %   \label{eq.1}
      $  rev(Q_1^*)=\sum_{k=1}^{r}\rho_{1,k}\leqslant r \cdot \max \{ \rho_1 \}$
%    \end{eqnarray}
    and
 %   \begin{eqnarray}
  %      \label{eq.2}
       $ rev(Q_2')=\max \{ \rho_1 \}.$
%    \end{eqnarray}
%    From (\ref{eq.1}) and (\ref{eq.2}), 
So we have
    $rev(Q_1^*)\leqslant r\cdot rev(Q_2').$
        %Similarly, as shown by the figure above,
    Similarly, we have
  %  \begin{eqnarray}
    %    \label{eq.3}
    $    rev(Q_i^*)\leqslant r\cdot rev(Q_{i+1}') \text{ for all } i=1,2,...,\mu -1.$
   % \end{eqnarray}
    Summing up %(\ref{eq.3}) 
    for all $i=1,2,...,\mu -1$, we know
    \begin{eqnarray}
        \label{eq.4}
        \sum_{i=1}^{\mu -1}rev(Q_i^*)\leqslant r\sum_{i=2}^{\mu}rev(Q_i').
    \end{eqnarray}
    From Lemmas~\ref{land-of-oz} and \ref{lem:conclude-sbpdb} of Appendix \ref{app:bipartite} 
    we know the right-hand-side of (\ref{eq.4}) is no more than the total revenue earned by $\sbp$ during all $\mu$ windows, therefore
    %\begin{eqnarray}
     %   \label{eq.5}
      $  \sum_{i=1}^{\mu -1}rev(Q_i^*)\leqslant r\sum_{i=1}^{\mu}rev(Q_i)$.
%    \end{eqnarray}
    Since the revenue earned by $\opt$ in the last (i.e. $\mu ^{th}$) window is bounded by a constant $c$, %(\ref{eq.5}) yields
    $\sum_{i=1}^{\mu}rev(Q_i^*)\leqslant r\sum_{i=1}^{\mu}rev(Q_i) + c.$  In other words,
    $rev(\opt)\leqslant r \cdot rev(\sbp) + c = \lceil1/k \rceil \cdot rev(\sbp) + c.$
    %Therefore $\sbp$ is $\lceil 1/k \rceil$-competitive under case 1.
    
   Case 2: There may be empty windows (i.e. windows where $\sbp$ serves nothing). 
    Let $w$ denote the last empty window that occurred during the entire time limit and let $\tau$ denote the start time of window $w$. We analyze the requests served before, during, and after $w$.
    \begin{itemize}
    \item  Before window $w$:
since $\sbp$ serves nothing during window $w$, we know that all requests released before time $\tau$ have been served by $\sbp$. Let $b$ denote the total revenue of these requests. We know that before $\tau$, $\opt$ could have earned revenue at most $b$. 
    \item During window $w$: $\opt$ earns revenue $\rho_{w,1}, \rho_{w,2},...,\rho_{w,r}$ and $\sbp$ earns nothing. 
    \item After window $w$: now we proceed by running $\sbp'$  which serves the maximum revenue request served in the previous window in the $\opt$ schedule. Similar to (\ref{eq.4}), we have
%    \begin{eqnarray}
 %       \label{eq.7}
        $rev(Q_i^*)\leqslant r\cdot rev(Q_{i+1}') \text{ for all }i=w,w+1,...,\mu -1.$
  %  \end{eqnarray}
    Summing up %(\ref{eq.7}) 
    for all $i=w,w+1,...,\mu -1$ yields
    %\begin{equation}
       % \label{eq.8}
      $  \sum_{i=w}^{\mu -1}rev(Q_i^*)\leqslant r\sum_{i=w+1}^{\mu }rev(Q_i').$
   % \end{equation}
    From Lemma~\ref{land-of-oz} (in the appendix) we know $r\sum_{i=w+1}^{\mu }rev(Q_i')$ %the right-hand-side of (\ref{eq.8}) 
    is no more than the total revenue earned by $\sbp$ during all windows after window $w$, therefore 
%    \begin{eqnarray}
 %       \label{eq.9}
   $     \sum_{i=w}^{\mu -1}rev(Q_i^*)\leqslant r\sum_{i=w+1}^{\mu }rev(Q_i).$
  %  \end{eqnarray}
    Since the revenue earned by $\opt$ in the last (ie. $\mu^{th}$) window is bounded by a constant $c$, we have
    \begin{equation}
        \label{eq.10}
        \sum_{i=w}^{\mu }rev(Q_i^*)\leqslant r\sum_{i=w+1}^{\mu }rev(Q_i) + c.
    \end{equation}
        \end{itemize}
    So
    \begin{equation}
        \label{eq.11}
        rev(\opt) = \sum_{i=1}^{\mu }rev(Q_i^*) \le b + \sum_{i=w}^{\mu }rev(Q_i^*) \le b+ r\sum_{i=w+1}^{\mu }rev(Q_i) + c.
    \end{equation}
  and
    \begin{eqnarray}
        \label{eq.12}
        rev(\sbp) = \sum_{i=1}^{\mu }rev(Q_i) = b +0+ \sum_{i=w+1}^{\mu }rev(Q_i).
    \end{eqnarray}
    %Eqns. (\ref{eq.11}) and (\ref{eq.12}) give the total revenue earned by $\opt$ and $\sbp$, respectively. 
%    From (\ref{eq.10}), we have
 %   \begin{eqnarray}
  %  \label{eq.12b}
   %     b+\sum_{i=w}^{\mu }rev(Q_i^*)\leqslant b+ r\sum_{i=w+1}^{\mu }rev(Q_i) + c.
%    \end{eqnarray}
    Combining \ref{eq.11} and \ref{eq.12} we have:
%    \begin{eqnarray*}
%        rev(\opt)\leqslant b+ r\sum_{i=w+1}^{\mu }rev(Q_i) + c
%    \end{eqnarray*}
   % \begin{eqnarray}
      %  \label{eq.13}
      %$  rev(\opt)-c\leqslant b+ r\sum_{i=w+1}^{\mu } rev(Q_i).$
   % \end{eqnarray}
    %Dividing by $rev(\sbp)$ on both sides %of (\ref{eq.13}), 
    %we get
    \begin{equation}
        \label{eq.14}
        \frac{rev(\opt) - c}{rev(\sbp)} \leqslant \frac{b+r\sum_{i=w+1}^{\mu }rev(Q_i)}{b+\sum_{i=w+1}^{\mu }rev(Q_i)}
%    \end{eqnarray}
 %   Consider the right-hand-side (RHS) of (\ref{eq.14}). We know that
  %  \begin{eqnarray}
   %     \label{eq.15}
       % \le \frac{b+r\sum_{i=w+1}^{\mu }rev(Q_i)}{b+\sum_{i=w+1}^{\mu }rev(Q_i)} 
       \leqslant \frac{rb+r\sum_{i=w+1}^{\mu }rev(Q_i)}{b+\sum_{i=w+1}^{\mu }rev(Q_i)} = r.
    \end{equation}
Which means   %$$\frac{rev(\opt)-c}{rev(\sbp)}\le r$$
 %      $$rev(\opt)\le r\cdot rev(\sbp)+c$$
       $rev(\opt)\le \lceil 1/k \rceil \cdot rev(\sbp)+c.$
%    Thus, $\sbp$ is $\lceil 1/k \rceil$-competitive under case 2.
   % \end{itemize}
  %  \end{enumerate}
   % Combining case 1 and case 2, we conclude that $\sbp$ algorithm is $\lceil 1/k \rceil$-competitive for nonuniform-revenue bipartite graph when edge weight has lower bound $kT/f$ and upper bound $T/f$.
\end{proof}

\newpage

\section{ Appendix}

In this section we provide all proofs missing from the main body of the paper. Most of these proofs were deferred to this section due strictly to space limitations.%, especially those in Section \ref{app:bipartite}. 
\subsection{ Proof from Preliminaries Section}\label{app:prelim}
\begin{proof}[Proof of Lemma \ref{lem:last2}]
    Consider the following instance for $f\geq 2$ for some non-preemptive deterministic algorithm \textsc{alg}. The adversary releases a request $(s, d, T-2T/f, 1)$ where the distance between $s$ and $d$ is $T/f$ and it takes $T/f$ time to travel between \textsc{alg}'s server location at time $T-2T/f$ and $s$. %(even if \textsc{alg}'s current location is between two nodes) \ccnote{I am finding this confusing again.  what does alg's current location refer to?  Does it refer to alg's location at time $T-2f$?  If so we should say that.} \adcnote{Ok, I added what you said, but I think we should take parenthetical stuff out - "even if alg's...}
    If \textsc{alg} chooses not to serve the request, then the adversary releases no more requests, so \textsc{alg} earns 0 while $\opt$ serves the request and earns 1. If \textsc{alg} serves the request, it moves to $s$ for $T/f$ time units, then serves the request until $T$, and earns revenue 1.  During this time, the adversary releases a request $(a, b, T-2T/f+\delta, k)$, for some small $\delta$ and an arbitrarily large $k$, where $a$ is the location of an optimal server, $\textsc{opt}$, at time $T-2T/f$. The $\textsc{opt}$ solution serves this request earning revenue $k$, so  $\frac{\textsc{opt}}{\textsc{alg}} = k$. 
    
    For the case of uniform revenues, we simply modify the % proof of Lemma~\ref{lem:last2} 
    above instance so that at time $T-2T/f + \delta$ for some small $\delta$ (i.e. while \textsc{alg} is serving the $(s, d, T-2T/f, 1)$ request), the adversary releases $k$ requests $r_1, r_2, \ldots, r_k$ such that the source of $r_1$ is the location of $\opt$ at time $T-2T/f$, the distance between the source and destination of each request is $\epsilon$ for some small $\epsilon$, each request has revenue 1, and the sequence of requests can be served in a chain (i.e. with no intermediary moves in between). The $\textsc{opt}$ solution serves the $k$ requests earning revenue $k$, so  $\frac{\textsc{opt}}{\textsc{alg}} = k$. 
\end{proof}

%%%%%%%%%%%%%%%%%%%%%%%%%%%%%%%%

%\adcnote{I wonder if we should mention that we assume every algorithm starts at the same origin and for this instance no requests are issued out of this location} \ccnote{the problem model as stated above implies that there is a single origin right?  i don't see a reason to say that no requests are issued from there right up front... does that matter?}. 
\begin{proof}[Proof of Theorem \ref{thm:genLB} for nonuniform revenues] %We start by giving the proof for the nonuniform-revenue case.
Consider the following instance with $f=5$ (so there are 5 time segments of length $T/f$). For simplicity, we let $X=T/f \ge 3$ denote the length of a time segment and therefore the maximum distance between two locations, so $T=5X$. All distances are $X$ unless otherwise stated. %Because $X\ge3$, the distances of $X-1$ and $X-2$ are above the minimum requirement of $1$. \ccnote{actually i don't think we have the min distance assumption in this work. i think that is something we have adopted in our latest discussions/research, but not part of this paper... i think there are instances we give in this paper where the min distance is less than 1 too, right?}
%Let 1 = 1, the minimum distance. 
Let $\opt$ denote an optimal schedule, let $\on$ denote a deterministic online algorithm and let $a_0$ denote the origin, i.e. the location of $\on$ and $\opt$ at time 0. The adversary releases requests $r_1 = (a_1, b_1, X, \epsilon)$ and  $r_2 = (a_2, b_2, X, \epsilon)$. Let $d(u, v)$ denote the distance between locations $u$ and $v$. % and let $d(b_1,c_2)=\1, d(c_2,d2)=X-\1, d(b_2,c_3)=\1$,  $d(c_3,d_3)=X-\1$, $d(a_3,b_3)= d(a_4,b_4)=X, d(b_3,c_4)=d(b_4,c_5)=X-2\1=X-2\1$, and $d(c_4,d_4)=d(c_5,d_5)=\1$.

\begin{enumerate}
    \item Case: $\on$ does not ever visit $a_1$ or $a_2$. Then the adversary releases no more requests, so $\on$ earns 0 while $\opt$ serves one of the requests and earns $\epsilon$ within the first $T-2T/f$.%\ccnote{should we say anything about why opt can't get credit for serving both?} \adcnote{added something}
    \item Case: $\on$ moves from its location at time $t_1 >X$ to either $a_1$ or $a_2$.
    Note $\on$ has not earned any revenue yet.  %The adversary will release a request $r_3$.
    \begin{enumerate}
        \item Case: $ X <t_1 < 2X$. 
        Since the following release time is after $t_1$, we may assume w.l.o.g. that $\on$ is at $a_1$ at time $t_1+X$. Then the adversary releases request $r_3=(b_2,c_2,2X,\epsilon)$. %\ccnote{should we say this WLOG part before saying the request that's released?  b/c the request that's released assumes alg goes to $a_1$ right?}
        When $\on$ arrives at $a_1$ there is fewer than $3X$ units of time remaining %\adcnote{actually, $\on$ doesn't have to serve $r_1$ here} 
        so %$\on$ serves $r_1$ and
        there is insufficient time for $\on$ to serve more than one request.% from $a_1$.
        %where $d(b_1,c_2)=X$.
%When $\on$ arrives at $a_2$ at time $t_1+X$, there is time $5X-(t_1+X) = 4X - t_1 < 3X$ remaining which is insufficient for $\on$ to serve any two of  $r_1, r_2, r_3$ so $\on$ earns revenue at most $\epsilon$. $\opt$ moves from the origin to $a_1$, serves $r_1$, moves to $c_2$ and serves $r_3$ in time $3X$ earning revenue $2\epsilon$, so $\frac{\opt}{\on} = 2$.
 \item Case: $ t_1 = 2X$. 
        Since the following release time is after $t_1$, we may assume w.l.o.g.  that $\on$ is at $a_1$ at time $t_1+X$. Then the adversary releases request $r_3=(c_2,d_2,2X+1,\epsilon)$,
 where $d(b_2,c_2)=1$ and $d(c_2,d_2)=X-1$.
        When $\on$ arrives at $a_1$, there is $2X$ time  remaining.  There is insufficient time for $\on$ to serve more than one request. 
         %there is time $5X-(t_1+X) = 2X$ remaining which is insufficient for $\on$ to serve two or more of $r_1, r_2 and r_3$ so $\on$ earns revenue at most $\epsilon$. $\opt$ moves from the origin to $a_1$, serves $r_1$, moves to $c_2$ and serves $r_3$ in time $3X$ earning revenue $2\epsilon$, so $\frac{\opt}{\on} = 2$.  
\item Case: $t_1 > 2X$. When $\on$ arrives at $a_1$ or $a_2$ at time
$t_1+X > 3X$, there is $< 2X$ time remaining. The adversary releases request $r_3=(c_2,d_2,2X+1,\epsilon)$,
 where $d(b_2,c_2)=1$ and $d(c_2,d_2)=X-1$. 
%Note that this case differs from case b since we cannot say
%without loss of generality that $\on$ is anywhere because the release time of $r_3$ could be before $t_1$.
It takes
at least $2X$ time for $\on$ to serve two requests from either $a_1$ or $a_2$ so there is insufficient time for $\on$ to serve more than one request.%; i
 % requires $2X$ to serve $r_3$. %When $\on$ arrives at $a_2$ at time $t1+X >3X$, there is time $5X-(t_1+X) <2X$ remaining which is insufficient for $\on$ to serve two or more of .

\end{enumerate}
In cases 2(a)-2(c), %when $\on$ arrives at $a_2$ %at time $t_1+X=3X$
there is %time $5X-(t_1+X) = 2X$ remaining which is 
insufficient time for $\on$ to serve two or more of $r_1, r_2$ and $r_3$ so $\on$ earns revenue at most $\epsilon$. On the other hand,
$\opt$ serves $r_2$ and $r_3$ by traversing
$a_0,a_2,b_2,c_2$ or $a_0,a_2,b_2,c_2,d_2$
in time $3X$ earning revenue $2\epsilon$, so $\frac{\opt}{\on} = 2$.

\item Case: $\on$ moves from its location at time $t_1 =X$ to either $a_1$ or $a_2$.
Since all future requests are released after time $X$, we can assume w.l.o.g. $\on$ moves to $a_1$ and arrives there at time $2X$. 
Then the adversary releases the requests:
  $r_3 = (a_3, b_3, X+ 1, 1)$,
  $r_4 = (a_4, b_4, X + 1, 1)$,
  and
  $r_5 = (c_3, d_3, 3X-1, 1)$ or $r_5=(c_4, d_4, 3X-1, 1)$, depending on when and where $\on$ moves.  
  Let $d(b_3,c_3)=d(b_4,c_4)=X-2$, and $d(c_3,d_3)=d(c_4,d_4)=1$.
 \begin{enumerate}
    \item Case: $\on$ serves $r_1$ or $r_2$. 
    Then there will be at most $2X$ time remaining. 
    So $\on$ can serve at most one additional request
    of revenue $1$ or $\epsilon$.
    So $\on$ earns at most $1 + \epsilon$. 
    \item Case: $\on$ does not serve $r_1$ or $r_2$ but  moves from $a_1$ at time $t_2$.
    Note that if $\on$ does not eventually move to one of $a_3, a_4, c_3, c_4$, then $\on$ would earn $0$.

    \begin{enumerate} 
    \item Case: $2X \le t_2 < 3X-1$. Since $r_5$ is released after $t_2$,
    we may assume w.l.o.g. that $\on$ will move to $a_3$ and $r_5 = (c_4, d_4, 3X-1, 1)$. %($\on$ does not know about $r_5$ at time $t_2$).
    When $\on$ arrives at  $a_3$ at time $t_2+X \ge 3X$, there is $\le2X$ time remaining.
    But it takes at least time $2X+1$ for $\on$ to earn revenue $2$ (e.g. by traversing $a_3, b_3, c_4, d_4$ with
    other paths taking longer).
    Hence $\on$ earns at most revenue $1$.

 \item Case: $t_2 \ge 3X-1$. 
 Then $\on$ sees which $r_5$ request was released and can choose to
 head towards any of the locations such as $a_3, a_4, c_3, c_4$,
 and arrives there at time $t_2+X \ge 4X-1$,
 so there is $X+1$ remaining.
 W.l.o.g let $r_5= (c_4, d_4, 3X-1, 1)$.
 From any location, it will take $\on$ at least time
 $2X-1$ to earn revenue $2$ (e.g. by traversing $a_4, b_4, c_4, d_4$ with other paths taking longer).
 Since $X\ge3$, we have $2X-1 > X+1$,
 and so $\on$ can earn at most $1$.
\end{enumerate}
In all subcases of Case 3,  $\on$ earns at most revenue $1+\epsilon$.
On the other hand, $\opt$ serves $r_4$ and $r_5$ by traversing $a_0$,$a_4$, waiting time $1$, and then
traversing $b_4, c_4, d_4$, in total time $3X$ earning revenue 2, so $\frac{\opt}{\on} = 2/(1+\epsilon)$.
\end{enumerate}
%Note X-2 is acceptable because X>=3, so X-2 is at least the minimal distance 1.
\end{enumerate}
%For the case of uniform revenues, see Theorem \ref{thm:glb} and its proof, below.
\end{proof}
%%%%%%%%%%%%%%%%%%%%%%%%%%%%%%%%

%\begin{theorem}\label{thm:glb} No deterministic online algorithm for OLDARP with uniform revenue can be guaranteed to earn more than twice the revenue earned by an optimal offline solution in the first $T-2T/f$ time units.
%\end{theorem}
\begin{proof}[Proof of Theorem \ref{thm:genLB} for uniform revenues] %\adcnote{I wonder if we should mention that we assume every algorithm starts at the same origin and for this instance no requests are issued out of this location} \ccnote{the problem model as stated above implies that there is a single origin right?  i don't see a reason to say that no requests are issued from there right up front... does that matter?}. 
Consider the following instance with $f=5$ (so there are 5 time segments of length $T/f$). For simplicity, we let $X=T/f$ denote the length of a time segment and therefore the maximum distance between two locations, so $T=5X$. All distances are $X$ unless otherwise stated.
We let the uniform revenue be $1$.
Fix a positive integer $k$ and let
$0<\delta<X/(2k)$.

Let $\opt$ denote an optimal schedule, let $\on$ denote a deterministic online algorithm and let $a_0$ denote the origin, i.e. the location of $\on$ and $\opt$ at time 0. The adversary releases requests $r_1 = (a_1, a_2, X, 1)$ and  $r_2 = (b_1, b_2, X, 1)$. Let $d(u, v)$ denote the distance between locations $u$ and $v$. % and let $d(b_1,c_2)=\1, d(c_2,d2)=X-\1, d(b_2,c_3)=\1$,  $d(c_3,d_3)=X-\1$, $d(a_3,b_3)= d(a_4,b_4)=X, d(b_3,c_4)=d(b_4,c_5)=X-2\1=X-2\1$, and $d(c_4,d_4)=d(c_5,d_5)=\1$.
\begin{enumerate}
    \item Case: $\on$ does not ever visit $a_1$ or $b_1$. Then the adversary releases no more requests, so $\on$ earns 0 while $\opt$ serves one of the requests and earns $1$ within the first $T-2T/f$. 
    \item Case: $\on$ moves from its location at time $t_1 >X$ to either $a_1$ or $b_1$.
    Note $\on$ has not earned any revenue yet.
    %The adversary will release a request $r_3$.
    \begin{enumerate}
        \item Case: $ X <t_1 < 2X$. 
        Since the following release time is after $t_1$, we may assume w.l.o.g. that $\on$ is at $a_1$ at time $t_1+X$; the adversary releases request $r_3=(b_2,b_3,2X,1)$. %\ccnote{should we say this WLOG part before saying the request that's released?  b/c the request that's released assumes alg goes to $a_1$ right?}
        When $\on$ arrives at $a_1$ there is $< 3X$ time  remaining
        so %$\on$ serves $r_1$ and
        there is insufficient time for $\on$ to serve more than one request.% from $a_1$.
        %where $d(b_1,c_2)=X$.
%When $\on$ arrives at $a_2$ at time $t_1+X$, there is time $5X-(t_1+X) = 4X - t_1 < 3X$ remaining which is insufficient for $\on$ to serve any two of  $r_1, r_2, r_3$ so $\on$ earns revenue at most $\epsilon$. $\opt$ moves from the origin to $a_1$, serves $r_1$, moves to $c_2$ and serves $r_3$ in time $3X$ earning revenue $2\epsilon$, so $\frac{\opt}{\on} = 2$.
 \item Case: $ t_1 = 2X$. 
        Since the following release time is after $t_1$, we may assume w.l.o.g.  that $\on$ is at $a_1$ at time $t_1+X$; the adversary releases request $r_3=(b_3,b_4,2X+\delta,1)$,
 where $d(b_2,b_3)=\delta$ and $d(b_3,b_4)=X-\delta$.
        When $\on$ arrives at $a_1$, there is $2X$ time  remaining.  There is insufficient time for $\on$ to serve more than one request. 
         %there is time $5X-(t_1+X) = 2X$ remaining which is insufficient for $\on$ to serve two or more of $r_1, r_2 and r_3$ so $\on$ earns revenue at most $\epsilon$. $\opt$ moves from the origin to $a_1$, serves $r_1$, moves to $c_2$ and serves $r_3$ in time $3X$ earning revenue $2\epsilon$, so $\frac{\opt}{\on} = 2$.  
\item Case: $t_1 > 2X$. When $\on$ arrives at $a_1$ or $b_1$ at time
$t_1+X > 3X$, there is $< 2X$ time remaining. The adversary releases request $r_3=(b_3,b_4,2X+\delta,1)$,
 where $d(b_2,b_3)=\delta$ and $d(b_3,b_4)=X-\delta$. 
%Note that this case differs from case b since we cannot say
%without loss of generality that $\on$ is anywhere because the release time of $r_3$ could be before $t_1$.
It takes
at least $2X$ time for $\on$ to serve two requests from either $a_1$ or $b_1$ so there is insufficient time for $\on$ to serve more than one request.
 % it requires $2X$ to serve $r_3$. %When $\on$ arrives at $a_2$ at time $t1+X >3X$, there is time $5X-(t_1+X) <2X$ remaining which is insufficient for $\on$ to serve two or more of .

\end{enumerate}
In cases 2(a)-2(c), %when $\on$ arrives at $a_2$ %at time $t_1+X=3X$
there is %time $5X-(t_1+X) = 2X$ remaining which is 
insufficient time for $\on$ to serve two or more of $r_1, r_2$ and $r_3$ so $\on$ earns revenue at most $1$. On the other hand,
$\opt$ serves $r_2$ and $r_3$ by traversing
$a_0,b_1,b_2,b_3$ (case 2(a))or $a_0,b_1,b_2,b_3,b_4$ (cases 2(b) and 2(c))
in time $3X$ earning revenue $2$, so $\frac{\opt}{\on} = 2$.

\item Case: $\on$ moves from its location at time $t_1 =X$ to either $a_1$ or $b_1$.
Since all future releases are released after time $X$, we can assume w.l.o.g. $\on$ moves to $a_1$ and arrives there at time $2X$. 
Then the adversary releases the requests:
  $r'_i = (c_{i-1}, c_{i}, X+\delta, 1)$,
  $r''_i = (d_{i-1}, d_{i}, X+\delta, 1)$,
  for $i=1,\ldots,k$,
  where
  $d(c_i,c_j)=d(d_i,d_j)=|i-j|X/k$.
  Depending on what $\on$ does,
  the adversary will also choose
  $w$ with $0<w\le X-\delta$ and release
  $\bar r_i = (e_{i-1}, e_{i}, 3X-w, 1)$ for $i=1,\ldots,k$,
  where
  $d(e_i,e_j)=|i-j|w/k$, and
  one of $d(e_0,c_k)=X-\delta-w$ or
  $d(e_0,d_k)=X-\delta-w$
  will be chosen.  
 \begin{enumerate}
    \item Case: $\on$ serves $r_1$ or $r_2$. 
    Then there will be at most $2X$ time remaining. 
    So $\on$ can serve at most $k$ additional requests from the $r',r'',\bar r$ family.
    So $\on$ earns at most $1 + k$. 
    \item Case: $\on$ does not serve $r_1$ or $r_2$ but  moves from $a_1$ at time $t_2$.
    Note that if $\on$ does not eventually move to one of $c_i,d_i,e_i$, then $\on$ would earn $0$.

    \begin{enumerate} 
    \item Case: $2X \le t_2 < 3X-\delta$. The adversary chooses $w=3X-t_2-\delta>0$ so that
    $3X-w = t_2+\delta$.
    Since $\bar r_i$ are released after $t_2$,
    we may assume w.l.o.g. that $\on$ will move to some $c_j$ and set
    $d(e_0,d_k)=X-\delta-w$
    and $d(e_0,c_k)=X$.
    When $\on$ arrives at  $c_j$ at time $t_2+X$, there is $4X-t_2$ time remaining.
    But it takes at least time $X/k+X+ w
    =4X-t_2-\delta+X/k>4X-t_2$ for $\on$ to earn revenue $k+1$ (e.g. by traversing $c_{k-1},c_k,e_0,e_1,
    \ldots,e_k$ with
    other paths taking longer).
    Hence $\on$ earns at most revenue $k$.

 \item Case: $t_2 \ge 3X-\delta$. 
 In this case,
 the adversary picks $w=\delta$ and
 releases the $\bar r_i$ requests at time $3X-\delta$
 and set $d(e_0,d_k)=X-\delta-w$
    and $d(e_0,c_k)=X$.
 Then $\on$  can choose to
 head towards any of the locations such as the sources of $r', r'',\bar r$,
 and arrives there at time $t_2+X \ge 4X-\delta$,
 so there is $X+\delta$ remaining.
 From any location, it will take $\on$ at least time
 $X/k+X-\delta-w+w$ to earn revenue $k+1$ (e.g. by traversing $d_{k-1},d_k, e_0,e_1,\ldots,e_k$ with other paths taking longer).
 But $X/k+X-\delta-w+w>X+\delta$ 
 because we chose $2\delta<X/k$.
 So $\on$ can earn at most $k$.
\end{enumerate}
In all subcases of Case 3,  $\on$ earns at most revenue $1+k$.
On the other hand, $\opt$ can traverse $a_0,d_0,d_1,\ldots,d_k,e_0,e_1,\ldots,e_k$, with pausing at $d_0$ until time $X+\delta$, in total time $3X$ to earn $2k$.
So $\frac{\opt}{\on} = 2k/(1+k)$.
\end{enumerate}
%Note X-2 is acceptable because X>=3, so X-2 is at least the minimal distance 1.
\end{enumerate}
\end{proof}

\subsection{ Proof of SBP lower bound} \label{app:lb}

We now present the formal proof of Theorem \ref{thm:lb} in Section \ref{sec:lb}.% We reproduce the figure here for the reader'sconvenience.

\begin{proof}[Proof of Theorem \ref{thm:lb}]
Consider the instance depicted in Figure \ref{fig:lb5}.
Fix $f >2$ to be an even integer.
Fix $T>0$, such that the time segment length $T/f>1$, where distances are discretized so 1 is the smallest possible unit of distance, i.e. all distances are integer-valued. Let $h=T/(2f)$.  Assume further $h>1$.
Let $0<\epsilon<1$ be vanishingly small and let $B>0$. 

Let $o$ be the origin, with other points in the metric space being $u_i$ for $i=1,2,\ldots,f$ and 
$v_i$ for $i=1,2,\ldots,m$,  where $m$ will be determined below.  

The idea is that $\sbp$ will take the path $o,u_1,\ldots,u_f$ in time $T$ serving a single request of revenue $B+\epsilon$ every other time segment as prescribed by the algorithm.
Meanwhile, discounting the revenue earned in the last two time segments, $\opt$ will take the path $o,v_1,\ldots,v_m, u_2,...,u_{f-2}$ in time $T-2T/f=T-4h$.
The distances are shown below each edge in the figure:
$d(o,u_1)=1$,
$d(u_i,u_{i+1}) = 1$ for $i$ odd, $d(u_i,u_{i+1}) =h$ for $i$ even,
$d(o,v_1)=1$,
$d(v_i,v_{i+1})=h+1$, for $i=1,\ldots,m-1$, 
$d(v_m, u_2)=h+1$, and all other distances (not shown) are $h$. %(We just need to make sure that a path such as u_2, u_3, u_5, u_6 has length greater than 2h.)

The requests are depicted as directed edges in the figure.  They are:
$(u_1,u_2, 0, \epsilon),\\
(u_{2i+1}, u_{2i+2}, 4ih, B+\epsilon)$ for $i=1,\ldots, f/2-1$,
$(u_{2i}, u_{2i+1}, 4ih+1, B)$ for $i=1,\ldots, f/2-1$,
$(v_i, v_{i+1}, 1, B)$ for $i=1,...,m-1$,
and $(v_m, u_2, 1, B)$.

Note that $\sbp$ will take the path $o,u_1,...,u_f$:\\ (1) At time $t=0$, $\sbp$ will choose drive $(o,u_1)$ followed by request $(u_1,u_2)$ because that is all that is available.\\
(2) For $k=1,\ldots,f/2-1$, at time $t = 4hk$ (the start time of a pair of time segments of length $2h=T/f$ each),
$\sbp$ is at vertex $u_{2k}$.
The available requests (that have not yet been served) are:
$(u_{2k+1}, u_{2k+2}, 4kh, B+\epsilon)$,
$(u_{2i}, u_{2i+1}, 4ih+1, B)$ for $i=1,..., k-1$, 
$(v_i, v_{i+1}, 1, B)$ for $i=1,...,m-1$, and $(v_m, u_2, 1, B)$.  Note that none of the requests of revenue $B$ along the top path arrive in time for $\sbp$ to serve more than a single request at a time.  Further, since we are looking for a revenue set that has path length at most $2h$, we cannot put together a path of length at most $2h$ that has 2 or more of these requests (since an edge from either $u_{2k+1}$ or $u_{2k+2}$ to any of the other vertices listed above has weight $h$ by design).
Thus a maximum revenue set chosen by $\sbp$ using a path of length at most $2h$ has only one request.
And a maximum revenue request would clearly be the request $(u_{2k+1}, u_{2k+2}, 4kh, B+\epsilon)$.
Thus $\sbp$ would drive $(u_{2k},u_{2k+1})$ at time $t=4kh$ followed by the request $(u_{2k+1},u_{2k+2})$ at time $t=4kh+2h$.
And  at time $t=4(k+1)h$, $\sbp$ would be at vertex $u_{2k+2}$.
%, completing the induction and the proof of the claim.

Thus $\sbp$ earns revenue $\epsilon + (f/2 - 1)(B +\epsilon) = B(f-2)/2 +f\cdot\epsilon/2$.

We now consider $\opt$, but we disregard the last two time segments.  That is, we consider only the revenue earned up until time $T - 2T/f = 2fh - 4h = (2f-4)h$.
Consider the following path, call it $\opt’$ for convenience:
    $o, v_1,\ldots,v_m, u_2,\ldots, u_{f-2}$. 
Note we stop at node $u_{f-2}$ because the requests from $u_{f-2}$ to $u_{f-1}$ and $u_{f-1}$ to $u_f$ are not yet released before time $(2f-4)h$.
%Check that the requests are released on or before when OPT’ gets to the source of the request.
%Calculate the revenue OPT’ makes, call the revenue by the name OPT’ also.
%Compute the ratio OPT’/ALG and take the limit as eps->0, del->0, T->infinity, while holding f fixed.
%In terms of $m$, 
$\opt’$ takes $1+m\cdot(h+1)$ time to get to $u_2$ and takes time $(1+h)\cdot(f/2 - 2)$ to go from $u_2$ to $u_{f-2}$. The total is
  $1 + m(h+1) + (1+h)\cdot(f/2 - 2)$.
%We want this to be <= (2f-4)h.
The largest $m$ for which $\opt'$ is completed before time $(2f-4)h$ is
%So we have:
% $$ 1 + m(h+1) + (f/2 -2)\cdot(1+h) \le (2f-4)h$$
 % $$m(h+1) <= (2f-4)h -1 - (f/2 - 2)*(1+h) $$
%  $$2m(h+1) <= (4f - 8)h -2 - (f-4)(1+h)$$
%  $$2m(h+1) <=4fh - 8h - 2 - (f - 4 + fh - 4h)$$
%  $$2m(h+1) <= 4hf - 8h - 2 -f + 4 - hf  +4h$$
 % $$2m(h+1) <= 3hf - 4h - f + 2$$
%  $$m \le (3hf - 4h - f + 2)/(2(h+1)) $$
%If we let %So we can take
  $$m = \lfloor (3hf - 4h - f + 2)/(2(h+1)) \rfloor.$$
  %, then 
  Observe that for this value of $m$, we have $1 + m(h+1) + (f/2 -2)\cdot(1+h) \le (2f-4)h$ as needed.
%Next, we check that OPT’ serves the requests as claimed.
Clearly, $\opt’$ can serve all the requests on the path $v_1,...,v_m,u_2$ because these requests were all released at time $t=1$.
Now, for each $k=1,...,f/2-2$, $\opt’$ arrives at vertex $u_{2k}$ at time 
    $\tau_k = 1+m(h+1)+(1+h)(k-1)$.
By Lemma \ref{lem:lb5} in Appendix \ref{app:lb}, $\tau_k \ge 4kh+1$ for each $k=1,\ldots,f/2-2$. 
Therefore the requests $(u_{2k}, u_{2k+1}, 4kh+1, B)$ 
and $(u_{2k+1}, u_{2k+2}, 4kh, B+\epsilon)$ are released on or before $\tau_k$, allowing $\opt’$ to serve these two requests when it reaches $u_{2k}$ at time $\tau_k$.
Therefore all the drives starting at $v_1$ are revenue generating requests for $\opt’$.
 
Now, $\opt’$ has revenue $mB$ from $v_1$ up to $u_2$ and revenue $(2B + \epsilon)(f-4)/2$ from $u_2$ to $u_{f-2}$. 
Let $rev(\textsc{alg})$ denote the total revenue earned by a schedule, $\textsc{alg}$. Then we can write %Call this total revenue OPT’.
$rev(\opt’) = mB + (2B+\epsilon)(f - 4)/2
= mB + Bf -4B + \epsilon(f-4)/2
= \lfloor (3hf - 4h - f + 2)/(2(h+1)) \rfloor B + fB -4B + \epsilon(f-4)/2$.

The ratio is thus
$$ \frac{rev(\opt’)}{rev(\sbp)} = \frac{ \lfloor (3hf - 4h - f + 2)/(2(h+1)) \rfloor B + fB -4B + \epsilon(f-4)/2 } { B\cdot(f-2)/2 +f\cdot\epsilon/2 }.$$

Taking the limit as $\epsilon$ approaches $0$, 
 \begin{eqnarray}\nonumber
     \lim_{\epsilon\rightarrow 0} \frac{rev(\opt’)}{rev(\sbp)}
  &=& \frac{ \lfloor (3hf - 4h - f + 2)/(2(h+1)) \rfloor B + fB -4B}{ B\cdot(f-2)/2}\\\nonumber
  &=& \frac{\lfloor (3hf - 4h - f + 2)/(2(h+1)) \rfloor  + f -4 }{  (f-2)/2 }
 \end{eqnarray}
Next we take the limit as $T$ approaches infinity, which is the same as taking $h$ to infinity, because $f$ is fixed.
The inside of the floor can be rewritten as

$$(3hf - 4h - f + 2)/(2(h+1))
  %= (3hf - 4h +3f - 4 - 4 f + 6)/(2(h+1)
  %= (   (3f - 4)(h+1)   -   (4f-6)     )/(2(h+1))
  %=   (3f - 4)/2      -     (4f-6)     )/(2(h+1))
  =   3f/2 - 2     -    (2f-3)/(h+1).$$

Note $3f/2-2$ is an integer since $f$ is even 
So when  $h+1>2f-3$, we have $0 < (2f-3)/(h+1) < 1$.
Thus
   $\lfloor   3f/2 - 2     -    (2f-3)/(h+1) \rfloor =  3f/2 - 2 - 1 = 3f/2 - 3$
when $h > 2f -4$.
Then
$  \lim_{T\rightarrow \infty}  \lfloor   3f/2 - 2     -    (2f-3)/(h+1) \rfloor =  3f/2 - 2 - 1 = 3f/2 - 3.$
Therefore
  $$\lim_{T\rightarrow \infty} \lim_{\epsilon \rightarrow 0} \frac{rev(\opt’)}{rev(\sbp)} 
  = \frac{ 3f/2 - 3  + f - 4 }{ (f-2)/2 }
  = \frac{ 5f/2 - 7 }{ (f-2)/2 } 
  =\frac{(5f -14)}{(f-2)}
  =5 - 4/(f-2)$$

Summarizing, for a fixed $f>2$, this instance gives a lower bound of
$5 - 4/(f-2)$
as $\epsilon$ approaches 0 and $T$ approaches infinity.
\end{proof}
For definitions and context for the following lemma, please refer to the proof of Theorem \ref{thm:lb} above, from which the following lemma is referenced.  Recall that in the analysis of the lower bound instance, $\tau_k$ was used to denote the time at which $\opt'$ arrives at vertex $u_{2k}$. 

\begin{lemma}\label{lem:lb5}
$\tau_k \ge 4kh+1$ for $k=1\ldots f/2-2$.
\end{lemma}
\begin{proof}Recall that for each $k=1 \ldots f/2-2$, $\opt’$ arrives at vertex $u_{2k}$ at time 
    $\tau_k = 1+m(h+1)+(1+h)(k-1)$.
    By definition of $m = \lfloor (3hf - 4h - f + 2)/(2(h+1)) \rfloor$,
we know $$m \ge (3hf - 4h - f + 2)/(2(h+1))  -1.$$
So
   $$m(h+1) \ge (3hf - 4h - f +2)/2 - (h+1).$$
Then
  $$\tau_k \ge 1 + (3hf - 4h - f +2)/2 - (h+1)+ (1+h)(k-1).$$
We then rewrite as
  \begin{align}
      \tau_k &\ge& 1 + (3f/2 - 2)h -f/2 + 1 - h - 1 + k-1+ h(k-1)  - 4kh + 4kh\\
  &\ge& 1 + 4kh + (3f/2-2 +k-1-1 -4k)h - f/2 -1 +k\\
  &\ge& 1+ 4kh + (3f/2 - 4 -3k)h - f/2- 1 + k\\
&\ge& 1 + 4kh + (3f/2 -4 -3k)h - f/2 + 4/3 + k   - 1 -4/3\\
  &\ge& 1 + 4kh + (3f/2 -4 -3k)(h - 1/3) - 1 - 4/3 
  \end{align}
 Since $k\le f/2-2$, then $3k\le 3f/2-6$.
Then $3f/2 - 6 -3k\ge 0$.
Hence $3f/2-4-3k \ge 2$.
So
 $$ (3f/2 -4 -3k)(h - 1/3) -1 - 4/3\\
   \ge 2(h - 1/3) - 2(h-1) - 4/3\\
   \ge 2h - 2/3  -2h + 2 - 4/3\\
   \ge 0$$
Thus we have shown
  $\tau_k \ge 1 + 4kh$.
\end{proof}

\subsection{ Proofs for the uniform revenue case}\label{app:uni-rev}
We begin by reitrating several definitions that will aid us in proving Lemma \ref{lemma-a-4} which is needed for Theorem \ref{4-comp-theorem} of Section \ref{sec:uniformrev}.
%\noindent
%\textbf{Part I: Partition the Set}
%As in the proof of Theorem~\ref{onvsopt}, 

%As in the proof of Theorem~\ref{onvsopt}, we consider a modified version of the $\sbp$ schedule, that we refer to as $\sbp'$, which serves exactly the same set of requests as $\sbp$, but does so one time window earlier.  
For all windows $i=1,2,..., m$, where $m=\lceil f/2\rceil -1$, we let $S'_i$ denote the set of requests served by $\sbp'$ in window $i$ and $S^*_i$ denote
the set of requests served by $\opt$ during the time segment of window $i$ with greater revenue, i.e. $S^*_i =\argmax \{ rev(\opt (t_{2i-1}), rev(\opt(t_{2i})) \}$ where $rev(\opt (t_j))$ denotes the revenue earned by $\opt$ in time segment $t_j$. 

We define a new set $J_i^*$ as the set of requests served by $\opt$ during the time segment of window $i$ with less revenue, i.e. $J_i^*=\argmin \{ rev(\opt (t_{2i-1}), rev(\opt(t_{2i})) \}$.

Let $S_i^*=A_i \cup X_i^* \cup Y_i^*$, and $S_i'=A_i \cup X_i \cup Y_i$, where:
\begin{enumerate}
    \item $A_i$ is the set of requests that appear in both $S_i^*$ and $S_i'$. 
    \item $X_i^*$ is the set of requests that appear in $S_w'$ for some $w=1,2,...,i-1$. Note there is only one possible $w$ for each individual request $r\in X_i^*$, because each request can only be served once. 
    \item $Y_i^*$ is the set of requests such that no request from $Y_i^*$ appears in $S_w'$ for any $w=1,2,...,i-1,i$. 
    \item $X_i$ is the set of requests that appear in $S_w^*$ for some $w=1,2,...,i-1$. Note there is only one possible $w$ for each individual request $r\in X_i$, because each request can only be served once.  
    \item $Y_i$ is the set of requests such that no request from $Y_i$ appears in $S_w^*$ for any $w=1,2,...,i-1,i$.
    \end{enumerate}
Note that elements in $Y_i$ can appear in a previous $J_w^*$ for any $w=1,2,...,i-1, i$ or in a future $S^*_v$ or $J^*_v$ for any $v=i+1,i+2,...,m$, or may not appear in any other sets.

Also note that since each request can be served at most once, we have:
\begin{enumerate}
    \item $A_1 \cap X_1^* \cap Y_1^* \cap A_2 \cap X_2^* \cap Y_2^*\cap...\cap A_{m}\cap X_{m}^*\cap Y_{m}^* = \emptyset$ 
    \item $A_1\cap X_1\cap Y_1\cap A_2\cap X_2\cap Y_2\cap ...\cap A_{m}\cap X_{m}\cap Y_{m} = \emptyset$
\end{enumerate}

Given the above definitions, we have the following lemmas: 

\begin{lemma}
    \label{lemma-a-1}
    All requests $r\in X_i^*$ must satisfy that $r \in Y_w$ for some $w=1,2,...,i-1$, and there is only one possible value of $w$.
\end{lemma}
\begin{proof}
    By definition, each request of $X_i^*$ must appear in $S_w'$ for some $w=1,2,...,i-1$, and there is only one possible value of $w$. Let $r$ be a request of $X_i^*$. 
    We know that $r$ must appear in either  $A_w$, or $X_w$, or $Y_w$.  However, $r$ cannot appear in $A_w$, since otherwise $r$ would have been served in $S_w^*$, where $w<i$, which is a contradiction since $r$ is served in $S_i^*$. Similarly, $r$ cannot appear in $X_w$, since otherwise $r$ would have been served in $S_v^*$ for some $v=1,2,...,w-1$, where $v<w<i$, which is a contradiction since we know $r$ is served in $S_i^*$. By elimination, $r$ must be a request of $Y_w$.
    %begin{itemize}
    %    \item $r$ cannot be in $A_w$, for otherwise $r$ would have been served in $S_w^*$, where $w<i$, which is a contradiction since we know $r$ is served in $S_i^*$.
      %  \item $r$ cannot be in $X_w$, for otherwise $r$ would have been served in $S_v^*$ for some $v=1,2,...,w-1$, where $v<w<i$, which is a contradiction since we know $r$ is served in $S_i^*$.  
%    \end{itemize}
%By elimination, $r$ can only be a request of $Y_w$. 
\end{proof}

\begin{lemma}
    \label{lemma-a-2}
    $X^*_1\cup X_2^* \cup ... \cup X_i^* \subseteq Y_1\cup Y_2\cup ...\cup Y_{i-1}$ for all $i=2,3,...,m$. 
\end{lemma}
\begin{proof}
    We prove this by induction. For the base case, by Lemma \ref{lemma-a-1},  $X_1^*$ must be a subset of $Y_0$, where $Y_0$ is the empty set; $X^*_2$ must be a subset of $Y_1$. Therefore, $X_1^*\cup X_2^*=\emptyset \cup X_2^* \subseteq Y_1$.
    For the inductive case, assume  \begin{eqnarray}
        \label{lemma-a-2-temp-1}
        X_1^*\cup X_2^* \cup ...\cup X_k^* \subseteq Y_1\cup Y_2\cup ...\cup Y_{k-1}.
    \end{eqnarray}
    Consider $X_{k+1}^*$ and $Y_k$. 
    By Lemma \ref{lemma-a-1}, elements of $X_{k+1}^*$ can come from only two sources: $Y_k$ and $Y_1\cup Y_2 \cup ...Y_{k-1}$.  
    Therefore, 
    \begin{eqnarray}
        \label{lemma-a-2-temp-2}
        X_{k+1}^* \subseteq Y_1\cup Y_2 \cup ...\cup Y_{k-1} \cup Y_k
    \end{eqnarray}
    Combining  (\ref{lemma-a-2-temp-1}) and (\ref{lemma-a-2-temp-2}), we have
    $$X_1^*\cup X_2^* \cup ...\cup X_k^* \cup X_{k+1}^* \subseteq Y_1\cup Y_2\cup ...\cup Y_{k-1} \cup Y_k.$$
\end{proof}

We now restate Lemma \ref{lemma-a-4}, originally presented in Section \ref{sec:uniformrev}, for the reader's convenience before giving its proof.
\begin{lemma}
  %  \label{lemma-a-4}
    $|X_1^*|+|X_2^*|+...+|X_i^*| \leqslant |Y_1|+|Y_2|+...+|Y_{i-1}| + |Y_i|$ for all $i=1,2,...,m$
\end{lemma}

\begin{proof}
    The size of a set $A$ must be at least the size of any subset of $A$ so from Lemma \ref{lemma-a-2}, we have for all $i=1,2,...,m$,
    $$|X_1^*\cup X_2^* \cup ...X_i^*| \leqslant |Y_1\cup Y_2\cup ...Y_{i-1}|$$
     
%    Due to the mutual exclusiveness between 
Recall we observed above that $X_1^* \cap X_2^* \cap ...\cap X_i^* = \emptyset$ and 
%as well as the mutual exclusiveness between 
$Y_1 \cap Y_2 \cap... \cap Y_{i-1} = \emptyset$.  Hence, we can rewrite the inequality as
    $$|X_1^*|+|X_2^*|+...+|X_i^*| \leqslant |Y_1|+|Y_2|+...+|Y_{i-1}|$$

    By adding $|Y_i|$ on the right-hand-side, we have proven the claim.%get for all $i=1,2,...,\lceil f/2 \rceil - 1$
%    $|X_1^*|+|X_2^*|+...+|X_i^*| \leqslant |Y_1|+|Y_2|+...+|Y_{i-1}| + |Y_i|$. 
\end{proof}

\subsection{ Proofs for bipartite graphs}\label{app:bipartite}

\begin{proof}[Proof of Theorem \ref{thm:roldarpb}]

 For every node $u$ in the RDARP graph, make two nodes $u_1$ and $u_2$ for the bipartite graph such that $u_1 \in V_1$ (the source side) and $u_2 \in V_2$ (the destination side), then set the weights $w(u_1,u_2) = w(u_2,u_1)$ in the bipartite graph to a diminishingly small $\epsilon > 0$.
    
    For every  edge $(u,v)$ in the RDARP graph with weight $w(u,v)$, make two edges in the bipartite graph: one from $u_1 \in V_1$ to $v_2 \in V_2$ and a second from $u_2 \in V_2$ to $v_1 \in V_1$, both with weight $w(u,v)$. %request in the bipartite graph from $u$ to $v'$ with weight $w(u,v)$. 
    %For each node $v_1 \in V_1$ and each node $v_2 \in V_2$ make edges in the bipartite graph from $v_2$ to $v_1$ with weight 0.
     Finally, for each request from source $s$ to destination $d$ in the RDARP instance, create an equivalent-revenue request in the RDARP-B instance from $s_1 \in V_1$ to $d_2 \in V_2$.  Set the time limit in the RDARP-B instance to be $T+\delta$, where $T$ is the time limit of the RDARP instance and $\delta = T\epsilon$ in order to accommodate any time needed to travel the $\epsilon$-weight edges.  We also choose $\epsilon$ so that $\delta$ is smaller than the smallest-weight edge of the original graph.
    
    Observe that a schedule $\sigma$ with revenue $R$ exists in the RDARP instance if and only if a schedule $\sigma'$ with revenue $R$ exists in the RDARP-B instance: 
    
    (1) Let the path $P = (a,b,c,\ldots)$ be the path followed in the execution of the schedule $\sigma$ in the RDARP instance.  The equivalent path in the RDARP-B instance would be $P' = (a_1, b_1, b_1, c_2, c_1, \ldots)$.  Notice that by the construction above, $P'$ has the same revenue and execution time (to within an additive $\delta$) in the bipartite instance as $P$ does in the original instance.
    
    (2) Any path $P'$ followed in the execution of a schedule $\sigma'$ in the bipartite graph must alternate from one side of the bipartition to the other. %between vertices from $V_1$ and vertices from $V_2$.  
    Empty moves in $P'$ from a vertex $x_1 \in V_1$ to another vertex $y_2 \in V_2$ in the bipartite graph represent empty moves from $x$ to $y$ in $P$.  Moves in $P'$ from any node $y_2 \in V_2$ to another node $z_1 \in V_1$ also represent empty moves from vertices $y$ to $z$ in $P$.
    Any moves in $P'$ along $\epsilon$-weight edges are not moves at all along the corresponding path $P$ in the original graph, as they take virtually no time and represent the server staying put in $P$.  Finally, by construction, any requests served in $P'$ will be along an edge from some vertex $x_1 \in V_1$ to another vertex $y_2 \in V_2$, which will correspond to a request with source $x$ and destination $y$ in the RDARP instance.  In sum, $P$ will take essentially the same amount of time (minus at most $\delta$ time units) and earn the same revenue as $P'$.  %  moves on non-request edges  have the form $(a_1,b_1,b_1,c_2,c_1,\ldots)$ since the only available edges from $V_2$ to $V_1$ are from each node in $V_2$ to its corresponding node in $V_1$...% Given a path $P$
    %by the Given a sequence $\sigma = (a,b,c,\ldots)$ of requests in the ROLDARP instance, %there can be an edge from $u$ to $v$ in the bipartite graph with the same weight.% and from $v$ to $u$. 
\end{proof}

% The following lemma is used in the proof of Theorem \ref{thm:r-uniform-comp} below. %in Section \ref{app:bipartite}.
% \begin{lemma}
% \label{ceil2x}
%     $\lceil \frac{\lceil 2x \rceil}{2} \rceil=\lceil x \rceil$ for any real number $x$. 
% \end{lemma}
% \begin{proof}
%     There are two cases: (1) $x$ is integer; (2) $x$ is not integer.
    
%     If $x$ is an integer, then the theorem is trivially true.
    
%     If $x$ is not integer, then let’s suppose $x = a + b$, where $a$ is an integer and $0 < b < 1$. Then we have $\lceil \frac{\lceil 2x \rceil}{2} \rceil = \lceil \frac{\lceil 2a + 2b \rceil}{2} \rceil = \lceil a + \frac{\lceil 2b \rceil}{2} \rceil$. We need to further consider the following two sub-cases: (i) $0< b \leqslant \frac{1}{2}$ ; (ii) $\frac{1}{2} < b < 1$. 
    
%     If sub-case (i) is true, then $\lceil 2b \rceil = 1$. Therefore, 
%     $\lceil \frac{\lceil 2x \rceil}{2} \rceil =  \lceil a + \frac{1}{2} \rceil = a + 1 =  \lceil a + b \rceil = \lceil x \rceil$.
%     %$LHS=RHS$ indicates that the theorem is true under sub-case (i). 
    
%     If sub-case (ii) is true, then $\lceil 2b \rceil = 2$. Therefore, 
%     $\lceil \frac{\lceil 2x \rceil}{2} \rceil =  \lceil a + \frac{2}{2} \rceil = \lceil a+1 \rceil = a + 1 =  \lceil a + b \rceil = \lceil x \rceil.$
%     %$LHS=RHS$ indicates that the theorem is true under sub-case (i). 
% %    Overall, $\lceil \frac{\lceil 2x \rceil}{2} \rceil=\lceil x \rceil$ must be true for all real number $x$.
% \end{proof}

%%%%%%%%%%%%%%%%%%%%%%%%%%%%%%%%%%%%%%%%%%%%%%%%%%%%%%%%%%%%%%%%%%%%%

%We now present the proof of 

\begin{proof}[Proof of Theorem \ref{thm:r-uniform-comp}]

 We define a \textit{k-unit} %\ccnote{since this is the min edge-weight, can we call it something different?  like maybe atomic-unit or something else?  i don't get the intuition behind the term cross-unit...} 
 as a length of time of duration $kT/f$. As before, we refer to each pair of consecutive time segments as a time window. We say a $k$-unit belongs to time window $i$ if it ends within time window $i$. Note that a $k$-unit may straddle two windows by starting at one window and ending in the next.

    If $2/k$ is an integer, then there are strictly $\frac{2T/f}{kT/f}=2/k$ $k$-units in one time window. Note that this is true even if the window contains a straddling $k$-unit, since this $k$-unit will force another one to straddle into the next time window. 
    
   If $2/k$ is not an integer, then there are $\lfloor \frac{2T/f}{kT/f} \rfloor=\lfloor 2/k \rfloor$ non-straddling $k$-units within one window.  If the window contains a straddling $k$-unit, the number of $k$-units will be $\lfloor 2/k \rfloor +1=\lceil 2/k \rceil$.  Among those $\lceil 2/k \rceil$ $k$-units, at most $\lceil \lceil 2/k \rceil /2 \rceil = \lceil 1/k \rceil$ % (please refer to Lemma~\ref{ceil2x} in the appendix for a proof of this equation) 
   of them can be used in an $\opt$ schedule to serve requests since no algorithm can serve two or more requests consecutively without a move in between. Therefore, whether or not $2/k$ is an integer, the maximum number of requests that can be served in each window is $\lceil 1/k\rceil$. 
    
    Let $\mu = \lceil f/2 \rceil$ denote the total number of windows. There are two cases based on the performance of $\sbp$: 
    \begin{enumerate}
    \item $\sbp$ serves at least 1 request per window. In this case, $\sbp$ serves at least $\mu$ requests, and $\opt$ serves at most $\mu \cdot \lceil 1/k \rceil$. Therefore, $\frac{rev(\opt)}{rev(\sbp)} \le \frac{\mu \lceil 1/k \rceil}{\mu} \le \lceil 1/k \rceil$. 
    
    \item There exists at least one window where $\sbp$ serves no requests. We refer to such a window as an ``empty window.'' Consider the last empty window that occurred within the entire time limit, and denote this window as $w$. Let $\tau$ denote the start time of window $w$.  We analyze the requests served (1) before, (2) during, and (3) after $w$:
  
    \begin{itemize}
    \item Before window $w$: Since $\sbp$ serves nothing during window $w$, we know that all requests released before time $\tau$ have been served by $\sbp$.  Let $b$ denote this number of requests. So before $\tau$, $\opt$ could have served at most $b$ requests. 
    
    \item During window $w$: $\opt$ can serve at most $\lceil 1/k \rceil$ requests and $\sbp$ serves no requests.
    
    \item After window $w$: Suppose there are $x$ windows after $w$ (excluding window $w$). By definition, window $w$ is the last window where $\sbp$ serves no requests, hence during the $x$ windows afterwards, $\sbp$ serves at least one request per window. On the other hand, $\opt$ serves at most $\lceil 1/k \rceil$ requests per window. 
    \end{itemize}
    We know that $\lceil 1/k \rceil \geqslant 1$.
    Therefore we have $rev(\opt) \le b+\lceil 1/k \rceil +x\cdot \lceil 1/k \rceil$ and $rev(\sbp)\ge b+0+x.$
    Hence
    $$rev(\opt)\leqslant \lceil 1/k \rceil \cdot rev(\sbp)+\lceil 1/k \rceil$$
      \end{enumerate}
\end{proof}

%%%%%%%%%%%%%%%%%%%%%%%%%%%%%%%%%%%%%%%%%%%%%%%%%%%%%%%%%%%%%%%%%%%%

The proof of Theorem \ref{thm:r-nonuniform-comp} (in Section \ref{sec:bipartite-nonuniform})
relies on the fact that $rev(\sbp)\geqslant rev(\sbp' )$ which we now prove using the following two lemmas. 
Recall that within the $i^{th}$ window $\sbp'$ serves exactly one request: the maximum revenue request served by $\opt$ during the $(i-1)^{th}$ window. We now create another subroutine, called $\sbpdb$, that, for every every window, serves the highest revenue available request. %In other words, $\sbpdb$ still greedily chooses one request to serve per window, but with less constraint than $\sbp'$. Then we show that $rev(\sbp)\geqslant rev(\sbpdb)\geqslant rev(\sbp')$. 
Lemma \ref{land-of-oz} shows why $rev(\sbpdb)\geqslant rev(\sbp')$, and Lemma \ref{lem:conclude-sbpdb} shows why $rev(\sbp)\geqslant rev(\sbpdb)$, so we know $rev(\sbp)\geqslant rev(\sbp')$.

\begin{lemma}
    \label{land-of-oz} 
    
    Let $O'$ and $O''$ denote the set of revenues of the $\mu = \lceil f/2 \rceil$ requests served by $\sbp'$ and $\sbpdb$ by the end of window $\mu$, sorted in descending order. 
    Let $O'[z]$ and $O''[z]$ denote the $z^{th}$ element of $O'$ and $O''$, respectively, where $z=1,2,...,\mu$.
   
    Then, $O''[z]\geqslant O'[z]$ for all $z=1,2,...,\mu$.  
\end{lemma}

\begin{proof}
    We proceed by strong induction.
        Recall $Q_i'$ is the set of revenues of all requests served by $\sbp'$  during window $i$. Let $Q_i''$ and $Q_i^*$ denote the set of revenues of all requests served by $\sbpdb$ and $\opt$, respectively, during window $i$.
         
    \vspace{2mm}
 
    \textbf{Base case.} $z=1$. 
    We know that $O'[1]$ = $\max \{ Q_1^*\cup Q_2^* \cup ...\cup Q_{\mu -1}^* \}$.
    
    Consider $\sbpdb$ during the last window; there are two cases: 1. $\sbpdb$ has served a request with revenue equal to or larger than $O'[1] = \max \{ Q_1^*\cup Q_2^* \cup ...\cup Q_{\mu -1}^* \}$; 2. $\sbpdb$ has not served such a request.
    
    In the first case, no matter which request $\sbpdb$ serves in the last window, $O''[1] \geqslant O'[1]$.
    
    In the second case, $\sbpdb$ will choose one available request with the maximum revenue to serve in the last window, so $Q_\mu''$ will have revenue either equal to $\max \{ Q_1^*\cup Q_2^* \cup ...\cup Q_{\mu -1}^* \}$ or larger than $\max \{ Q_1^*\cup Q_2^* \cup ...\cup Q_{\mu -1}^* \}$. %This is owing to the fact that all the available requests of $\sbpdb$ at this stage are constituted of not only requests that correspond to ($Q_1^*\cup Q_2^* \cup ...\cup Q_{\mu -1}^* \backslash \text{ \{requests that both $\opt$ and $\sbpdb$ have served before the last window\}}$) but also requests that $\opt$ did not choose to serve. 
    Thus, when $\sbpdb$ is done, $O''[1] \geqslant O'[1]$.

    %By definition of $O'$, $$O'[1]=\max \{ Q_1'\cup Q_2' \cup ...\cup Q_{\mu}' \}$$
    
    %By definition of $O$, $$O[1]=\max \{ Q_1 \cup Q_2 \cup ...\cup Q_{\mu} \}$$
    
    %Also, $$O[1]=\max \{ Q'_1 \cup Q'_2 \cup ...\cup Q'_{\mu} \cup \{ \text{revenues of requests served by $\sbpdb$ but not $\sbp'$ } \} \}$$
    
    \vspace{2mm}
 
    \textbf{Inductive case.} Suppose $O''[z]\geqslant O'[z]$ is true for $z=1,2,...,l$. Consider $z=l+1$. We will show by contradiction that $O''[l+1]\geqslant O'[l+1]$. 
Suppose $O''[l+1]<O'[l+1]$.

    By the definition of $O''$, we know $O''[1], O''[2], ..., O''[l]$ are the $l$ largest revenues served by $\sbpdb$, and 
    \begin{eqnarray}
        \label{oz-eq-1}
        O''[1]\geqslant O''[2]\geqslant ...\geqslant O''[l]
    \end{eqnarray}
    
    From the inductive hypothesis, 
    \begin{eqnarray}
        \label{oz-eq-2}
        O''[l]\geqslant O'[l]
    \end{eqnarray}
    
    Given the ordered nature of $O'$, we have
    \begin{eqnarray}
        \label{oz-eq-3}
        O'[l]\geqslant O'[l+1]
    \end{eqnarray}
    
    Given $O''[l+1]<O'[l+1]$, (\ref{oz-eq-2}), and (\ref{oz-eq-3}), we have
    \begin{eqnarray}
        \label{oz-eq-4}
        O''[l]\geqslant O'[l] \geqslant O'[l+1] > O''[l+1]
    \end{eqnarray}
    
    The general approach of this proof is that, if the request that corresponds to the revenue $O'[l+1]$ is not the request that corresponds to any of $O''[1], O''[2], ..., O''[l]$, then the $(l+1)^{th}$ largest request selected by $\sbpdb$ would have been $O'[l+1]$ instead of $O''[l+1]$, since $O''[l+1]<O'[l+1]$. This is a contradiction. Now we must affirm that the request corresponding to $O'[l+1]$ does not correspond to any of $O''[1], O''[2], ..., O''[l]$. To verify this precondition, we consider the two possible cases of (\ref{oz-eq-4}). 

    % ***************************************
    % Lambus types before here
        
    \vspace{2mm}
 
    \textbf{Case 1.}    
    If $O''[l] > O’[l]$ or $O’[l] > O’[l+1]$, then given (\ref{oz-eq-2}) and (\ref{oz-eq-3}), the consequence of this case will be $O[l]>O’[l+1]$. More generally, taking into account (\ref{oz-eq-1}) and (\ref{oz-eq-4}), we have
    \begin{eqnarray}
        \label{new_eq.5}
        O''[1]\geqslant O''[2]\geqslant ...\geqslant O''[l]>O’[l+1]>O''[l+1],
    \end{eqnarray}
    which tells us that revenue $O’[l+1]$ does not correspond to any of $O''[1], O''[2], ..., O''[l]$. Therefore $O’[l+1]$ should have been the $(l+1)^{th}$ largest value of the set $O''$ instead of $O''[l+1]$, which is a contradiction. 
        
    \vspace{2mm}
 
    \textbf{Case 2a.}
    If $O''[l]=O’[l]=O’[l+1] > O''[l+1]$ and the request that corresponds to $O’[l+1]$ never corresponds to any revenue among $O''[1], O''[2], ..., O''[l]$, then $O’[l+1]$ should also have been the $(l+1)^{th}$ largest value of the set $O''$ instead of $O''[l+1]$, which is a contradiction. 
        
    \vspace{2mm}
 
    \textbf{Case 2b.}
    If $O''[l]=O'[l]=O'[l+1] > O''[l+1]$, and the request that corresponds to $O’[l+1]$ also corresponds to $O''[g]$ for some $1\leqslant g \leqslant l$, which indicates that 
     \begin{eqnarray}
        \label{new_eq.6}
        O''[g]=O''[g+1]=…...=O''[l] > O''[l+1]
    \end{eqnarray}
    From the ordered nature of $O’$ and the inductive hypothesis, we know
    \begin{eqnarray}
        \label{new_eq.7}
        O''[g]\geqslant O’[g]\geqslant O'[g+1]\geqslant... \geqslant O’[l+1]
    \end{eqnarray}
    Combined with the fact that $O''[g]=O'[l+1]$, it must be the case that
    \begin{eqnarray}
    \label{new_eq.8}
    O''[g] = O'[g] = O'[g+1] =... = O'[l+1]
    \end{eqnarray}
    Given that there could exist elements in $O''$ prior to $O''[g]$ that are equal to $O’[l+1]$, we let $O''[x]$ be the first element in $O''$ that is equal to $O’[l+1]$, where $1\leqslant x\leqslant g$. 
    Similar to (\ref{new_eq.7}), we have 
    \begin{eqnarray}
    \label{new_eq.9}
    O''[x]\geqslant O'[x]\geqslant O'[x+1]\geqslant... \geqslant O’[l+1]
    \end{eqnarray}
    Combining (\ref{new_eq.9}) and $O[x]=O'[l+1]$, we know that 
    \begin{eqnarray}
    \label{new_eq.10}
    O''[x] = O'[x] = O'[x+1] =... = O'[l+1]
    \end{eqnarray}
    
    Now we know that in $O''$, there are $(l-x+1)$ elements that have revenue $O'[l+1]$, while in $O'$, there are $((l+1)-x+1)=(l-x+2)$ elements with revenue $O’[l+1]$. The set $O'$ has one more element with revenue $O'[l+1]$ than the set $O''$. 

    This is a contradiction because then $\sbpdb$ would have chosen the extra request in $O’$ that has revenue $O’[l+1]$ to serve instead of directly serving $O''[l+1]$.  
\end{proof}

    The direct consequence of Lemma \ref{land-of-oz} is that 
    $$rev(\sbpdb)= \sum_{z=1}^{n}O''[z] \geqslant \sum_{z=1}^{n}O'[z]=rev(\sbp' ).$$
    %$$\sum_{z=1}^{n}O[z] \geqslant \sum_{z=1}^{s}O'[z]$$
    %$$rev(\sbp)\geqslant rev(\sbp' )$$
    
\begin{lemma}
    \label{lem:conclude-sbpdb}
    Let $\sbpdb[i]$ and $\sbp[i]$ denote the $i^{th}$ window served by $\sbpdb$ and $\sbp$ respectively. Define $U=\{1,2,\dots,\mu\}$ as the set of all possible indices of windows. Define $A_j$ and $B_j$ as any two subsets of $U$ that satisfy the following criteria: 
    \begin{enumerate}
        \item Both $A_j$ and $B_j$ are of size $j$ and are in increasing order. 
        \item $A_j[m] \leqslant B_j[m]$ for all $m=1,2,\dots ,j$, where $A_j[m]$ is the $m^{th}$ element of $A_j$ and $B_j[m]$ is the $m^{th}$ element of $B_j$.
        \item For requests served in $ \{\sbpdb[A_j[1]], \sbpdb[A_j[2]] \dots \sbpdb[A_j[j]] \}$, if they are ever served by $\sbp$, then they are served only in $\{ \sbp[B_j[1]], \sbp[B_j[2]] \dots \sbp[B_j[j]] \}$.
    \end{enumerate}
    
    Then, for all such possible $A_j$ and $B_j$, we have
    $$\sum_{m=1}^{j}rev(\sbpdb[A_j[m]])\leqslant \sum_{m=1}^{j}rev(\sbp[B_j[m]])$$
\end{lemma}

\begin{proof}
    \textbf{Base case.} When $j = 1$, we assume $A_1[1] \leqslant B_1[1]$ and that if the request served by $\sbpdb[A_1[1]]$ is served by $\sbp$, it is served only in $\sbp[B_1[1]]$. This implies that the request served in $\sbpdb[A_1[1]]$ is available to $\sbp[B_1[1]]$, so given the greedy nature of $\sbp$, $rev(\sbp[B_1[1]]) \geqslant rev(\sbpdb[A_1[1]])$. Consequently, we can say $$\sum_{m=1}^{1}rev(\sbpdb[A_1[m]])\leqslant \sum_{m=1}^{1}rev(\sbp[B_1[m]])$$
    
    \textbf{Inductive Case.} We assume that $\sum_{m=1}^{j}rev(\sbpdb[A_j[m]])\leqslant \sum_{m=1}^{j}rev(\sbp[B_j[m]])$ for all $j$ where $1 \leqslant j \leqslant k$. Then we want to prove for all possible $A_{k+1}$ and $B_{k+1}$, 
    \begin{eqnarray}
        \label{evil-eq-1}
        \sum_{m=1}^{k+1}rev(\sbpdb[A_{k+1}[m]])\leqslant \sum_{m=1}^{k+1}rev(\sbp[B_{k+1}[m]])
    \end{eqnarray}
    
    There are two cases: 
    \begin{enumerate}
        \item $rev(\sbpdb[A_{k+1}[k+1]]) \leqslant rev(\sbp[B_{k+1}[k+1]])$
        \item $rev(\sbpdb[A_{k+1}[k+1]]) > rev(\sbp[B_{k+1}[k+1]])$
    \end{enumerate}
    
    If the first case is true, then combining it with the inductive hypothesis, (\ref{evil-eq-1}) is clearly true. %The inductive proof is then completed.  
    
    If the second case is true, we denote $r$ as the request served in window $\sbpdb[A_{k+1}[k+1]]$. Then the equation of case 2 can be rewritten as 
    \begin{eqnarray}
        \label{evil-case-2}
        rev(\sbp[B_{k+1}[k+1]])<rev(r)
    \end{eqnarray}
    
    Eqn. (\ref{evil-case-2}) implies that $r$ is not available to $\sbp[B_{k+1}[k+1]]$, because otherwise in window $B_{k+1}[k+1]$ $\sbp$ would have served some request(s) whose total revenue is at least the revenue of $r$. Suppose it is at window $\sbp[B_{k+1}[h]]$ that $\sbp$ serves $r$ where $ 1 \leqslant h < k+1$. 
    
    Define new subsets $A_k$ and $C_k$ where $A_k=A_{k+1}\backslash \{ A_{k+1}[k+1] \}$ and $C_k=B_{k+1} \backslash \{ B_{k+1}[h] \}$. It is evident that both $A_k$ and $C_k$ are of size $k$. Define $A_k[m]$ as the $m^{th}$ element of the newly defined $A_k$, and $C_k[m]$ as the $m^{th}$ element of the newly defined $C_k$. 
    
    Here we declare a shortcut of notation. For any set $V$ and integers $s\leqslant t$, $V[s:t]=\{ V[s],V[s+1],\dots,V[t-1],V[t] \}$. 
    
    %need to fix
    One observation is that for each element $C_{k}[m]$ of $C_{k}$, we have $C_{k}[m] \geqslant A_{k}[m]$. This is because each element of $C_{k}$ will be minimized when $C_{k}[1]=B_{k+1}[1]$, $C_{k}[2]=B_{k+1}[2]$, $\dots$, $C_{k}[k]=B_{k+1}[k]$ respectively. Since we know $B_{k+1}[m] \geqslant A_{k+1}[m]$, we have $C_{k}[m] \geqslant A_{k}[m]$. %We also know that in the normal case, each $C_{k}[m]$ does not have to correspond to $B_{k+1}[m]$. In this scenario, it should be that $C_{k}[m] > B_{k+1}[m]$. 
    %Therefore, we can deduce that $C_{k}[m] \geqslant A_{k}[m]$. 
    There are two sub-cases:
    
    \begin{enumerate}

    \item If $\sbp[B_{k+1}[h]]$ does not contain any of the requests of $\sbpdb[A_k[1:k]]$: Then $A_k$ and $C_k$ satisfy the three criteria listed in the lemma, so according to the inductive hypothesis, 
    \begin{eqnarray}
        \label{evil-eq-4}
        \sum_{m=1}^{k}rev(\sbpdb[A_{k}[m]])\leqslant \sum_{m=1}^{k}rev(\sbp[C_{k}[m]])
    \end{eqnarray}
    
    Adding $rev(r)$ on both sides of (\ref{evil-eq-4}), we have
    \begin{eqnarray}
        \label{evil-eq-5}
        \sum_{m=1}^{k+1}rev(\sbpdb[A_{k+1}[m]])\leqslant \sum_{m=1}^{k}rev(\sbp[C_{k}[m]])+rev(r) \leqslant \sum_{m=1}^{k+1}rev(\sbp[B_{k+1}[m]])
    \end{eqnarray}

    Note the second $\leqslant$ sign is valid since all the requests corresponding to $\sum_{m=1}^{k}rev(\sbp[C_{k}[m]])+rev(r)$ are served in $\sbp$ (and $\sbp$ may serve additional requests as well). 
    
    %(\ref{evil-eq-5}) shows that the inductive hypothesis is still true for $j=k+1$. 
    
    \item If $\sbp[B_{k+1}[h]]$ contains any of the requests of $\sbpdb[A_k[1:k]]$: Then $A_k$ and $C_k$ violate criterion (3) (since $B_{k+1}[h]$ is not in $C_k$) so the inductive hypothesis cannot be applied directly. Suppose there are $n$ such requests where $n\leqslant k$, and denote the total revenue of those $n$ requests as $N$. %Then we this problem can be solved by disregarding those $n$ requests in windows $A_{k}[1:k]$. 
    Then we define $A_{k-n}=A_{k} \backslash \{ \text{indices of windows where the $n$ requests reside in $\sbpdb$} \}$. We also shrink $C_{k}$ to $C_{k-n}$ by removing $n$ windows that do not contain any of the requests of $A_{k-n}$. Then we can follow the same reasoning above to deduce $C_{k-n}[m] \geqslant A_{k-n}[m]$ for all $1\leqslant m\leqslant k-n$.
    
    So according to the inductive hypothesis,
    \begin{eqnarray}
        \label{evil-eq-6}
        \sum_{m=1}^{k-n}rev(\sbpdb[A_{k-n}[m]])\leqslant \sum_{m=1}^{k-n}rev(\sbp[C_{k-n}[m]])
    \end{eqnarray}
    
    Adding $N$ and $rev(r)$ on both sides, we have
    \begin{eqnarray}
        \label{evil-eq-7}
        \sum_{m=1}^{k+1}rev(\sbpdb[A_{k+1}[m]])\leqslant \sum_{m=1}^{k-n}rev(\sbp[C_{k-n}[m]])+ N + rev(r) \leqslant \sum_{m=1}^{k+1}rev(\sbp[B_{k+1}[m]])
    \end{eqnarray}
    
    Note the second $\leqslant$ sign of (\ref{evil-eq-7}) is valid since all the requests corresponding to $\sum_{m=1}^{k-n}rev(\sbp[C_{k-n}[m]])+ N + rev(r)$ are served in $\sbp$ (and $\sbp$ may serve additional requests as well). 

   % Therefore, (\ref{evil-eq-7}) shows that the inductive hypothesis is still true for $j=k+1$. 
\end{enumerate}

    Finally, to prove $rev(\sbp)\geqslant rev(\sbpdb)$, we let $j=\mu$, $A_{\mu} = \{ 1,2,\dots ,\mu \}$, $B_{\mu} = \{ 1,2,\dots ,\mu \}$. This satisfies (1) $A_{\mu}$ and $B_{\mu}$ are in increasing order, (2) $A_{\mu}[m] \leqslant B_{\mu}[m]$ for all $m=1,2,\dots,\mu$, and (3) for requests served in $ \{\sbpdb[A_{\mu}[1]], \sbpdb[A_{\mu}[2]] \dots \sbpdb[A_{\mu}[\mu ]] \}$, if they are ever served by $\sbp$, are served only in windows $\{ \sbp[B_{\mu}[1]], \sbp[B_{\mu}[2]] \dots \sbp[B_{\mu}[\mu ]] \}$. %Since all the key criteria listed in Lemma \ref{lem:conclude-sbpdb} are fulfilled,
    Therefore, $rev(\sbp)\geqslant rev(\sbpdb)$ is simply a specific case of $\sum_{m=1}^{j}rev(\sbpdb[A_j[m]])\leqslant \sum_{m=1}^{j}rev(\sbp[B_j[m]])$ where $j=\mu$. 

\end{proof}
\bibliography{waoa19}

\end{document}